\documentclass{lmcs} 
\pdfoutput=1

\usepackage[utf8]{inputenc}
\usepackage{lastpage}
\lmcsdoi{17}{1}{19}
\lmcsheading{}{\pageref{LastPage}}{}{}%
{Oct.~14,~2019}{Mar.~02,~2021}{}


\usepackage{amsmath}
\usepackage{amsthm}
\usepackage{amssymb}
\usepackage{amsfonts}
\usepackage{MnSymbol}
\usepackage{stmaryrd}
\usepackage{MnSymbol}
\usepackage{url}
\theoremstyle{plain} 

\renewcommand{\above}{\sqsupseteq}
\newcommand{\incl}{\hookrightarrow}
\newcommand{\upup}{\twoheaduparrow}
\newcommand{\fromabove}{\searrow}

\newcommand{\Int}{\mathbb{Z}}
\newcommand{\lo}{\underline}
\newcommand{\hi}{\overline}
\newcommand{\pref}{\mathsf{pref}}
\newcommand{\lift}{_\bot}

\newcommand{\entails}{\vdash}
\newcommand{\wex}{\tilde\exists}

\newcommand{\Cons}{\mathsf{Cons}}
\newcommand{\step}{\mathrel{\searrow}}

\newcommand{\IQ}{\mathsf{I}\mathbb{Q}}
\newcommand{\waybelow}{\ll}

\newcommand{\SO}{\mathcal{O}} 

\newcommand{\inv}{^{-1}}

\newcommand{\sqsup}{\bigsqcup}
\newcommand{\below}{\sqsubseteq}
\newcommand{\Nat}{\mathbb{N}}
\newcommand{\Rat}{\mathbb{Q}}
\newcommand{\Real}{\mathbb{R}}
\newcommand{\IR}{\mathsf{I}\mathbb{R}}
\newcommand{\sep}{.\,}

\theoremstyle{plain}
\newtheorem{lemma}[thm]{Lemma}

\newtheorem{propn}[thm]{Proposition}

\theoremstyle{definition}
\newtheorem{defn}[thm]{Definition}
\newtheorem{example}[thm]{Example}

\newtheorem{remark}[thm]{Remark}
\newtheorem{comparison}[thm]{Comparison}

\begin{document}

\title{Constructive Domains with Classical Witnesses}


\author[D.~Pattinson]{Dirk Pattinson}	
\address{The Australian National University}	
\email{dirk.pattinson@anu.edu.au}  

\author[M.~Mohammadian]{Mina Mohammadian}	
\address{The Australian National University and University of Tabriz}	
\email{mina.mohammadian@anu.edu.au}  






\begin{abstract}
 We develop a constructive theory of continuous domains from the
  perspective of program extraction. Our goal that programs 
  represent (provably correct) computation without witnesses of
  correctness is achieved by formulating correctness
  assertions classically. Technically, we start from a predomain
  base and construct a completion. We then investigate continuity
  with respect to the Scott topology, and present a construction of
  the function space. We then discuss our main motivating example in
  detail, and instantiate our theory to real numbers that we
  conceptualise as the total elements of the completion of the
  predomain of rational intervals, and prove a representation
  theorem that precisely delineates the class of representable continuous
  functions.
\end{abstract}

\maketitle

\section*{Introduction}\label{S:one}

The ability to extract programs from proofs is one of the hallmark
features of constructive mathematics \cite{Troelstra:1988:CMI}: from a proof of a
formula of the form $\forall x \exists y P(x, y)$ we can
automatically obtain a (computable) \emph{function} $f$ such that
$P(x, f(x))$ for all $x$. Within mathematics, the
variables usually have \emph{types}, such as natural or real
numbers, or functions between types. 

Computationally, while some of these types, such as the natural
numbers,  can be computed with directly, there is no immediate way
to compute with others.
The prime example
here are the real numbers that are represented either as
infinitely-long running Turing machines \cite{Weihrauch:2000:CA},
rational Cauchy sequences with modulus \cite{Bishop:1985:CA}, linear fractional
transformations \cite{Heckmann:2002:CRN}, digit streams
\cite{Marcial-Romero:2007:SSL} or domains \cite{Edalat:1998:DTA}. 

From the view of program extraction, the data structure that is used
to represent mathematical objects is systematically derived from
their definition. If we define real numbers to be Cauchy sequences
with modulus, then programs extracted from an existence proof will
produce just that -- a Cauchy sequence with modulus. 
The vast majority of the work on constructive real analysis and
program extraction has focussed on the Cauchy representation and
its variants such as the signed digit representation, e.g.
\cite{Troelstra:1988:CMI,Bishop:1985:CA,Berger:2011:CPE}.  There is
little work on other representations, with the notable exception of
\cite{Bauer:2009:CTC} which develops a theory of constructive
domains that is instantiated to obtain representations of real numbers.

In domain theory, real numbers are represented as nested sequences
of (rational, or dyadic) intervals, with the interpretation that every interval gives an
upper and lower bound to the number being approximated. In other
words, every sequence element gives a guaranteed enclosure of the
actual result, and the successive computation of sequence elements
can be halted if the actual precision, measured by the interval
width, falls below a given threshold. 

Compared with a representation as Cauchy sequences with modulus,
domains offer two attractive features. First, every stage of
approximation carries an actual error bound, rather than the
worst case error, as given by the modulus of convergence for Cauchy
reals. 
For example, computing the square root of $2$ using Newton iteration
(as carried out in e.g. \cite{Schwichtenberg:2016:CAW}) one obtains
a rational Cauchy Sequence $(a_n)_n$ such that for example, $|a_5 - a_n| \leq \frac{1}{6
\cdot 2^4} \approx 5.2  \times 10^{-3}$ for all $n \geq 5$.
Instantiating the same method to obtain a shrinking sequence $(a_n,
b_n)_n$ of nested rational intervals such that $a_n^2 \leq 2 \leq
b_n^2$ one obtains that $|b_n - a_n| \approx 5.7 \times 10^{-49}$
for $n \geq 5$.
Both methods use the same initial approximation of $\sqrt{2}$, and
indeed the computed Cauchy sequence is identical to the sequence of
upper interval endpoints.
The significant difference is explained as the modulus is a worst
case estimate, whereas the differences between upper and lower
interval endpoint are obtained from the actual computation and avoid
over-estimation. We showcase this by means of an example in Section
\ref{sec:sqrt}.

The second attractive feature of a domain theoretic approach is that  most classes of domains are closed under the
formation of function spaces, i.e. one systematically obtains a
representation of the space of e.g. real-valued functions. 

Both motivate the development of a more general theory of domains,
as e.g. carried out in \cite{Bauer:2009:CTC}. Our work is similar in
spirit, focuses on extracted programs and data type as an end goal.
Specifically, our aim is to extract 
(necessarily effective) functions that operate on the basis of the
domains under consideration. For the special case of real numbers
(and functions), our goal is to obtain algorithms in the style
descibed in \cite{Edalat:1998:DTA}. There, mathematical operations
(such as computing square roots) are first extended to an
appropriate domain (such as the interval domain), then restricted to
the base of the domain, and in a third step, shown to be recursive
by considering a computable enumeration of basis elements. Indeed,
one of our goals is to short-circuit effectivity considerations that
are often laborious and provide little insight. Our slogan is
``proofs, not programs'' as the constructive reasoning (via a
realisability interpretation) immediately yields necessarily
recursive algorithms.

Putting the extracted algorithm into the centre of attention gauges
the formulation of the notion of domain, and this is where
differences to \cite{Bauer:2009:CTC} begin to emerge. The programs
we are seeking to extract should embody \emph{just} the computational
essence, but no additional terms that evidence correctness. For
example, when extracting a program to compute a real number, we only
seek a nested sequence of intervals, but \emph{not} a witness of the
fact that the intervals are converging to zero in width. This is
similar to the approach taken in \cite{Berger:2011:CPE} where
one freely adds (true) axioms without computational content to the
theory that forms the basis of extraction. That is, we are
interested in constructive existence, but are content with classical
correctness. Conceptually, this can be understood as phrasing
correctness in the classical (double negation) fragment of
constructive logic. Technically,  the (intended) consequence of this
is that correctness proofs do not have any computational content,
and are therefore invisible after program extraction, using a
standard realisability  interpretation~\cite{Troelstra:1988:CMI}. 

For example, subjecting the proof of the existence of the square
root of two to a realisability interpretation, our aim is to
extract only a nested sequence of intervals. To achieve this, the
definition of equality needs to be free of computational content.
We solve this by judiciously setting up the theory in such a way
that treats existence of objects as constructive existence, whereas
properties are usually formulated classically. 

Another aspect where our theory puts the extracted algorithm into the
centre of attention is the definition of completion of domain bases. 
It is one of the hallmark features of domains  that ideal elements (such
as infinite sequences, or real numbers) can be approximated by
elements of a \emph{base}. Constructively, we take the notion of a
base as primitive, and recover ideal elements in the completion of
the base. Classical domain theory, see e.g.~\cite{Abramsky:1994:DT}
usually considers completion by directed suprema. Here, take the
same approach as \cite{Bauer:2009:CTC} and consider completions by
infinite sequences, as they are much more easily representable
computationally.

\emph{Plan of the paper and main results.} We introduce the notion
of a \emph{predomain base} that is similar to \cite{Bauer:2009:CTC}
in Section \ref{sec:bases},
but our definition of the way-below relation is classical, and we
establish some basic lemmas, notably interpolation, for later use.
We also introduce our main motivating, and running, example, the
predomain base of formal intervals. 
In Section \ref{sec:completion}, we introduce the
continuous completion of a predomain base, along with a (defined)
notion of equality. As foreshadowed in the introduction, equality
(defined in terms of way-below) is classical and devoid of
computational content. The main result here is the extension property
that allows us to extend any continuous function defined on a
predomain base to its completion. 

In Section \ref{sec:topology} we align the order-theoretic notion of order-theoretic
continuity to topological continuity. As expected, this necessitates
a classical definition of the Scott topology which we also show to
be generated by upsets of the way-below relation as in the classical
theory. In particular, we can show that order-theoretic and
topological continuity coincide.
Our consideration of continuity naturally leads to the construction of
function spaces that we carry out in Section
\ref{sec:function-spaces}. In the classical theory, function
spaces are constructed as the set of Scott continuous functions,
with pointwise ordering. Here, we investigate the construction of
function spaces on the level of predomain bases. More specifically,
we present a construction of a predomain base, the completion of which
precisely captures the space of continuous functions between the
completion of two bases. In Sections \ref{sec:interval-domain} we
specialise our theory to our initial motivating example, and
recapture real numbers as the total elements of the (continuous
completion of the) domain of formal intervals. We show that the
Euclidean topology arises as the restriction of the Scott topology
to real numbers, and investigate the relationship between Cauchy
reals and the domain-theoretic reals. As a consequence of our
constructive existence -- classical correctness approach, both
notions are only equivalent if Markov's principle holds (and in
fact, we can prove Markov's principle from their equivalence). 
We conclude by relating $\epsilon$-$\delta$ continuous functions to
the restrictions of Scott continuous total functions. This
unearthes a new notion of continuity which appears to be weaker than
uniform continuity but at the same time stronger than pointwise
continuity that we call \emph{intensional non-discontinuity}. We
leave the question of a more detailed analysis of this notion to
future work.

\emph{Related work.} We have already mentioned \cite{Bauer:2009:CTC}
which is closest to the work reported in this paper. The main
differences are that our notions of way-below and equality are
defined classically whereas \emph{op.cit.} employs constructive
definitions. We also present a construction of function spaces as
completion of predomain bases in Section \ref{sec:function-spaces}. 

Our work stands in the tradition of Bishop-style constructive
analysis, \cite{Bishop:1985:CA}, and indeed we work in a purely
constructive setting. What is different is our treatment of real
numbers that we derive from the interval domain, similarly to the
classical 
treatment of real analysis in \cite{diGianantonio:1997:RNC,Edalat:1999:DTA} via continuous
domains, except that we
do not focus on the (classical) notion of computability. Again from a
classical perspective, our real numbers (and functions) can be thought of as the
total objects of (constructively understood) domains, studied in
\cite{Berger:1993:TSO}, although we don't investigate the notion of
totality \emph{per se}. The comparison between different notions of
continuity on the induced set of real numbers is of course insipred
by \cite{Ishihara:1992:CPC}.
Much of this paper is owed to discussions with Helmut
Schwichtenberg. His notes \cite{Schwichtenberg:2016:CAW} develop
constructive analysis with a view to program extraction, and the
question that motivated the present paper was whether this is also
possible using a domain representation of the reals, rather than a
Cauchy sequence representation with a worst-case modulus of
convergence. 

\section{Preliminaries and Notation} \label{sec:prelims}

We work in standard Bishop-style constructive mathematics 
\cite{Bishop:1985:CA} that we envisage as being formalised in
higher-type intuitionistic arithmetic $\mathsf{HA}^\omega$
\cite{Troelstra:1988:CMI}.  
We write $\Nat$ for the natural numbers, $\Int$ for the integers and
$\Rat$ for the rationals, and $\Rat_{>0}$ for the positive
rationals. 

We use the term
'weak existence' to refer  to the weak existential quantifier $\wex =
\neg\forall\neg$ which is constructively equivalent  to $\neg \neg \exists$. In informal reasoning, we often say that
`there must exist $x$ such that $A$' or `there weakly exists $x$
such that $A$' for $\wex x\sep A$. We read defined operations
universally, that is assuming that $A(x)$ defines $x$ uniquely
($(A(x) \wedge A(y)) \to x = y$) and we let $\phi(x)$ denote `the
unique $x$ such that $A(x)$', we read a formula $B(\phi(x))$ as
$\forall x\sep A(x) \to B(x)$. In particular, if there must exist
$x$ such that $A(x)$, using $\phi(x)$ does not assert (strong)
existence.

\section{Predomain Bases and Interpolation}  \label{sec:bases}

A \emph{predomain base} is a countable  ordered structure that collects
finitely representable objects used to approximate elements of ideal structures,
such as the real numbers.  Examples of predomain bases are finite
sequences (approximating infinite streams) and rational intervals
(approximating real numbers). The order structure captures
information content, such as the prefix ordering for finite
sequences, and reverse inclusion for rational intervals. 

Predomain bases are the constructive analogue of a base in classical
domain theory \cite{Abramsky:1994:DT}, where  arbitrary elements of
the domain can be displayed as directed suprema of base elements.
In a constructive setting, the totality of the domain is not
given and needs to be constructed, similar to the (constructive)
notion of real numbers as rational Cauchy sequences with a modulus
of convergence. This section discusses basic properties of predomain
bases, and we then construct completions in Section
\ref{sec:completion}.

\begin{defn}[Predomain Bases] \label{defn:predomain-bases}
Let $(B, \below)$ be a poset. A \emph{chain} in $B$ is a sequence
$(b_n)_n$ such that $b_n \below b_{n+1}$ for all $n \in \Nat$. 
If $(C, \below)$ is (another) poset, we call a monotone function $f: C \to D$
\emph{Scott continuous} if $f(\bigsqcup_n b_n) = \sqsup_n f(b_n)$
for all chains $(b_n)_n$, provided that all suprema in the last
equality exist.
An
element $b$ is \emph{way below} an element $c \in B$ if there must exist
$n \in \Nat$ such that $b \below x_n$ whenever $(x_n)$ is a chain in
$B$ with $\sqsup_n x_n \in B$ and $c \below \sqsup_n x_n$. 
We write $b \waybelow c$ if $b$ is way
below $c$, and also say that $b$ \emph{approximates} $c$. 
A chain $(b_n)$ is an \emph{approximating sequence} of $b \in B$ if
$b_n \waybelow b$ for all $n \in \Nat$ and $\sqsup_n b_n = b$. 
A \emph{predomain base} is a countable poset $B = \lbrace b_n \mid n
\in \Nat \rbrace$ with decidable ordering $\below$ in which every element
has   an approximating sequence. 

A non-empty, finite set $B_0 \subseteq B$
is \emph{consistent}, written $\Cons(B_0)$  if it must have an upper
bound, i.e. there must exist $b' \in B$ such that $b \below b'$ for all
$b \in B$. We say that consistency is \emph{continuous} if $a_i =
\sqsup_j a_{i, j}$ for all $i \in I$ and $\Cons \lbrace a_{i, j}
\mid i \in I \rbrace$ for all $j \in \Nat$ implies $\Cons \lbrace
a_i \mid i \in I \rbrace$ where $I$ is a nonempty, finite set. The  poset $(B, \below)$ is \emph{bounded complete} if every finite
consistent 
subset $B_0 \subseteq B$ has a least upper bound $\sqsup B_0$, and
\emph{pointed} if it has a least element $\bot \in B$. \end{defn}

\noindent
Note that all non-empty bounded complete posets are necessarily
pointed, and that consistency is not necessarily continuous as we
demonstrate in Example \ref{consis}. 

\begin{remark}
The notion of predomain base differs from that of
\cite{Bauer:2009:CTC} in that \emph{op.cit.} requires that an
approximating sequence be a $\waybelow$-chain. 
This immediately entails interpolation: if $x \waybelow z$ in a
predomain base (where every element has an approximating
$\waybelow$-chain), we have $z = \sqsup_n z_n$ for a
$\waybelow$-chain $(z_n)_n$ so that $x \below z_n$ for some $n$ by
definition of $\waybelow$. But then $x \below z_n \waybelow z_{n+1}
\waybelow z_{n+2} \below z$ so that $x \waybelow z_{n+1} \waybelow
z$, i.e. $y = z_{n+1}$ interpolates between $x$ and $z$. 

We require that every element $x \in B$ can be displayed as $x =
\sqsup_n x_n$ where each $x_n \waybelow x$ which is strictly weaker. 
As a consequence, we need additional hypotheses to establish
interpolation in Corollary \ref{cor:interpolation}. On the other
hand, our definition makes it easier to construct predomain bases as
we don't need to ensure that approximating sequences are
$\waybelow$-chains, as for example in the construction of function
spaces given later in Lemma \ref{lemma:function-space-predomain-base}.

We are also adopting a different (weaker) definition of the
way-below relation that is formulated using strong existence in
\emph{op.cit.} Both are equivalent if Markov's Principle is
assumed. By directly phrasing the way-below relation in terms of
weak existence, Markov's Principle can be avoided. 
A helpful pattern of proof that exploits weak existence is the
following. Suppose that $\Gamma, \exists x. A \entails B$ and
$\Gamma \entails \wex x. A$. Then $\Gamma \entails \neg\neg B$. 
Similarly, the notion of bounded completeness, phrased in terms of
weak existence, is stronger than that of \emph{op.cit.} which uses
strong existence. Technically, we need to use weak existence of an
upper bound to establish that the continuous completion of a
(bounded complete) predomain base has suprema of all increasing
chains (Corollary \ref{cor:bounded-complete-complete}).
Conceptually, weak existence suffices as the witness of boundedness
of a finite subset of a predomain base is not used in the
construction of the least upper bound.
\end{remark}

\begin{example} \label{example:bases}
  Let $B$ be a countable set with decidable equality, that is,
  $\forall b, b' \in B. (b =
  b') \lor \neg (b = b')$ is (constructively) provable. Then $(B, =)$
  and $(B^*, \below_\pref)$ are predomain bases where $B^*$ is the
  set of finite sequences of $B$ and $\below_\pref$ is the prefix
  ordering. Both are bounded complete and satisfy $x \waybelow x$
  for all $x \in B$ (resp. $x \in B^*$).

  If $(B, \below)$ and $(C, \below)$ are predomain bases, then so
  are $B \times C$ and $B + C$ with the pointwise and co-pointwise
  ordering. Moreover $(B\lift, \below\lift)$ is a predomain base
  where $B\lift = B \cup \lbrace \bot \rbrace$ (we tacitly assume
  $\bot \notin B$) and $b \below c$ if either $b = \bot$ or $b
  \neq \bot \neq c$ and $b \below c$. The predomain bases $B \times
  C$, $B + C$ and $B\lift$ are the \emph{product}, \emph{coproduct}
  and \emph{lifting} of $B$ and $C$ (resp. of $B$).
\end{example}

\begin{example} \label{ex:interval-predomain}
  The poset $\IQ = \lbrace (p, q) \in \Rat \times \Rat \mid
  p \leq q \rbrace$ ordered by $(p, q) \below (p', q')$ iff $p \leq
  p' \leq q' \leq q$ is called the \emph{predomain base of rational
  intervals}. We usually write $[p, q]$ for the pair $(p, q) \in
  \IQ$ and think of $[p, q]$ as a rational interval. For $\alpha  =
  [a, b] \in
  \IQ$, we sometimes  write $\alpha = [\lo \alpha, \hi \alpha]$
  to denote the lower and upper endpoint of $\alpha$, and $\alpha
  \pm \delta = [\lo \alpha - \delta, \hi \alpha + \delta]$ for the
  symmetric extension of $\alpha$ by $\delta \in \Rat_{\geq 0}$.
\end{example}

\noindent
It is not immediate (but easy) to see that $\IQ$ is a
predomain base.  The negative formulation of $\waybelow$ gives the
following characterisation that has been established in
\cite[Proposition 7.2]{Bauer:2009:CTC} using Markov's Principle.

\begin{lemma}\label{lemma:IQ-waybelow}
  Let $[p, q], [p', q'] \in \IQ$. Then
  $[p, q] \waybelow [p', q'] \mbox{ iff } p < p' \leq q' < q$.
\end{lemma}
\begin{proof}
  For the only-if direction, assume that $[p, q] \waybelow [p',
  q']$. As $[p', q'] = \sqsup_n [p' - 2^{-n}, q' + 2^{-n}]$ there
  must exist 
  $n \in \Nat$ such that $[p, q] \below [p' - 2^{-n}, q' +
  2^{-n}]$ from which we obtain that $p \leq p' - 2^{-n} < p' \leq q'
  < q' + 2^{-n} \leq q$ and hence $p < p' \leq q' < q$
 using decidability of order on $\Rat$.

  For the converse, assume that $p < p' \leq q' < q$ and assume that
  $[p', q'] \below [a, b] = \sqsup_n [a_n, b_n]$. Then $a = \sup_n
  a_n$ and $b = \inf_n b_n$. We claim that there must exist $n$
  and $m$ so that $a_n \geq p$ and $b_m \leq q$. So assume that $a_n
  \leq p$ for all $n \in \Nat$. Then $p$ is an upper bound of
  $(a_n)_n$ and therefore $a \leq p$. Hence $a \leq p < p' \leq a$,
  contradiction. The proof of classical existence of $m$ is analogous.
  Hence there must exist
  $N= \max \lbrace n, m \rbrace$ such that we have $p \leq a_N \leq
  b_N \leq q$, that is, $[p, q] \below [a_N, b_N]$. 
\end{proof}

\begin{lemma}
$\IQ$ is a predomain base. 
\end{lemma}
\begin{proof}
Let $[p, q] \in \IQ$ be given. Then $([p-2^{-n}, q+2^{-n}])_n$ is a
approximating sequence of $[p, q]$. 
\end{proof}

\begin{lemma} \label{lemma:IQ-cons-cont}
Consistency on $\IQ$ is continuous. 
\end{lemma}
\begin{proof}
Let $I$ be a finite set and let $\alpha_i = \sqsup_j \alpha_{i, j}
\in \IQ$ for all $i \in I$. Assume furthermore that $\lbrace
\alpha_{i, j} \mid i \in I \rbrace$ is consistent for all $j \in
\Nat$, we show that $\lbrace \alpha_i \mid  i \in I \rbrace$ is
consistent. 
The latter is the case if $\max \lbrace \underline{\alpha}_i \mid i
\in I \rbrace \leq \min \lbrace \overline{\alpha}_i \mid i \in I
\rbrace$. We have, for all $i \in J$, that $\max \lbrace
\underline{\alpha}_{i, j} \mid i \in I \rbrace \leq \min \lbrace
\overline{\alpha}_{i, j} \mid i \in I \rbrace$ which implies the
claim.
\end{proof}

\begin{example}\label{consis}
Consistency is not automatically continuous. Consider for instance
the predomain base $B = \IQ \setminus \lbrace [0, 0] \rbrace$ and
two sequences
$\alpha_n = [-1, 2^{-n}]$ and $\beta_n = [-2^{-n}, 1]$. Then
$\alpha_n$ and $\beta_n$ are consistent for all $n \in \Nat$ but
$[-1, 0] = \sqsup_n \alpha_n$ and $[0, 1] = \sqsup_n \beta_n$ are
not.
\end{example}

\noindent
We collect some basic facts about posets and the way-below relation
that are used in the proof of our first technical result, the
interpolation property (Proposition \ref{propn:interpolation} and Corollary
\ref{cor:interpolation}).  The majority of results are standard in
(classical) domain theory, see e.g. \cite{Abramsky:1994:DT}, and we
include them here both to be self-contained and to demonstrate that
they continue to hold in our framework.

\begin{lemma}\label{lemma:way-below-below}
  Let $(P, \below)$ be a poset for which $\below$ is decidable. Then $b \below c$ whenever $b
  \waybelow c$. Moreover, $a \below b \waybelow c$ implies
  that $a \waybelow c$, and similarly $a \waybelow b \below c$
  implies that $a \waybelow c$, for $a, b, c \in P$. 
\end{lemma}
\begin{proof}
  For the first item, assume that $b \waybelow c$. As $c = \sqsup_n
  c$ there must exist $n$ such that $b \below c$ whence $b \below
  c$.

  Now suppose that $a \below b \waybelow c$. Let $c \below \sqsup_n
  c_n$. Then there is $n$, weakly, such that $b \below c_n$ whence
  $a \below c_n$, too. Now suppose that $a \waybelow b \below c$ and
  let $c \below \sqsup_n c_n$. Then $b \below \sqsup_n c_n$ whence
  there is $n$, weakly, such that $a \below c_n$. 
\end{proof}

\noindent
The proof of the above lemma uses that $\below$ is $\neg\neg$-closed which follows from
decidability.

\begin{lemma}\label{lemma:waybelow-sup}
  Let $(P, \below)$ be a poset, $I$ be a finite set and $b_i, b_i' \in
  P$ with $b_i \waybelow b_i'$ for all $i \in I$. If $s = \sqsup_i b_i$ and $s' = \sqsup_i b_i'$ then $s
  \waybelow s'$.
\end{lemma}
\begin{proof}
  Let $(x_n)$ be a chain in P where $s' =\sqsup_i b_i' \below \sqsup_n x_n$. Since $\sqsup_i b_i'$
  is an upper bound of $\lbrace b_i' \mid i \in I \rbrace$, we have $b_i' \below \sqsup_n x_n$  for all
  $i \in I$.  Moreover, by assumption $b_i
  \waybelow b_i'$ for all $i \in I$, there must exist $n_i$ such that $b_i \below x_{n_i}$.
  Now by setting  $n  = \max \lbrace n_i \mid i \in I \rbrace$, then we have
  $b_i \below x_n$ for all $i \in I$. Hence $\sqsup_i b_i \below
  x_n$ as $\sqsup_i b_i$ is the least upper bound. 
\end{proof}
\begin{cor}\label{cor:waybelow-sup-bounded-complete}
Let $(B, \below)$ be a bounded complete poset and $I$ a
finite set. If $b_i, b_i' \in B$ for $i \in I$ such that $b_i
\waybelow b_i'$ and $\lbrace b_i' \mid i \in I \rbrace$ is
consistent, then both $\sqsup_i b_i$ and $\sqsup_i b_i'$ exist in
$B$ and $\sqsup_i b_i \waybelow \sqsup_i b_i'$.
\end{cor}
\begin{proof}
  As $\lbrace b_i' \mid i \in I \rbrace$ is consistent, $\sqsup_i
  b_i'$ exists in $B$, and  there must exist an upper bound
  $x \in B$ such that $b_i' \below x$ for all $i \in I$. 
  Moreover, 
  $b_i \below b_i' \below x$ for all $i \in I$ as $b_i
  \waybelow b_i'$ hence $\lbrace b_i \mid i \in I \rbrace$ is also
  consistent (with upper bound $x$), hence $\sqsup_i b_i$ exists in
  $B$, and the claim follows from Lemma \ref{lemma:waybelow-sup}.
\end{proof}

\begin{lemma}\label{lemma:sup-chain}
  Let $(P, \below)$ be a poset and $(a_n)_n$ and $(b_n)_n$
  chains in $B$. If $s = (\sqsup_n a_n) \sqsup (\sqsup_n b_n)$
  exists in $P$ and
  $s_n = a_n \sqsup b_n$ then $s = \sqsup_n s_n$.
\end{lemma}

\begin{proof}
  We first show that $s = a \sqsup b$ is an upper bound of $a_n \sqsup
  b_n$ for all $n \in \Nat$. We have that $a_n \below a \below a
  \sqsup b$ as $a$ is an upper bound of $a_n$ and $a \sqsup b$ is an
  upper bound of $b$. Similarly $b_n \below a \sqsup b$. Hence $a
  \sqsup b$ is an upper bound of $a_n \sqsup b_n$ and $a_n \sqsup
  b_n \below a \sqsup b$ follows as $a_n \sqsup b_n$ is the least
  upper bound. 

  Now we show that $a \sqsup b$ is indeed the least upper bound of
  $a_n \sqsup b_n$. So take another upper bound $x$, that is, $a_n
  \sqsup b_n \below x$ for all $n \in \Nat$. Then $a_n \below a_n
  \sqsup b_n \below x$ and $b_n \below a_n \sqsup b_n \below x$ for
  all $n \in \Nat$. Hence $a = \sqsup_n a_n \below x$ and $b  =
  \sqsup b_n \below x$. As $a \sqsup b$ is the least upper bound of
  $a$ and $b$, it follows that $a \sqsup b \below x$.
\end{proof}

\noindent
The above facts are used to prove our first result, the
(weak) interpolation property. 

\begin{propn} \label{propn:interpolation}
  Let $(B, \below)$ be a bounded complete predomain base,
  and assume  that $\waybelow$ on $B$ is decidable. 
  Then $B$ has the weak interpolation property, that is,
  whenever $x \waybelow z$ for $x, z \in B$ there must exist $y \in B$
  such that $x \waybelow y \waybelow z$.
\end{propn}

\begin{proof}
  We adapt the (classical) proof based on directed suprema (\cite[Lemma
  2.2.15]{Abramsky:1994:DT}) to our setting. Assume that $x
  \waybelow z$, and let $B = \lbrace b_n \mid n \in \Nat \rbrace$.

  As $z \in B$ and $B$ is a predomain base,
  we can find an element $b_n \in B$ with $b_n \waybelow z$ (e.g. the first
  element of an approximating sequence of $z$). By the same
  reasoning, we can find $b_m \in B$ with $b_m \waybelow b_n$. Let
  $o = \max \lbrace n, m \rbrace$ and consider the sequence
  \[ c_n = \sqsup \lbrace b_i \mid 0 \leq i \leq o+n, \exists 0 \leq j
  \leq o+n \sep b_i \waybelow b_j \waybelow z \rbrace. \]
  Then $c_n$ is well-defined, as suprema are taken over a non-empty,
  bounded (by $z$) and finite set. 

  We now claim that $\sqsup_n c_n = z$. First, it is clear that $c_n
  \below z$ for all $n \in \Nat$. To see that $z$ is a least upper
  bound of the $c_n$, suppose that $c_n \below u$ for all $n \in
  \Nat$, and we show that $z \below u$. Let $(z_n)_n$ be an
  approximating sequence for $z$. As $z = \sqsup_n z_n$, it suffices
  to show that $z_n \below u$ for all $n \in \Nat$. So let $n \in
  \Nat$. As $z_n \in B$, there exists an approximating sequence
  $(z_n^k)_k$ for $z$, and in particular $z_n^k \waybelow z_n
  \waybelow z$ for all $k \in \Nat$. Now fix an arbitrary $k \in
  \Nat$, we show that $z_n^k \below u$. As $B$ is countable, we can find $p, q \in \Nat$
  such that $z_n = b_p$ and $z_n^k = b_q$. Let $r = \max \lbrace p,
  q \rbrace$. Then $b_q = z_n^k  \waybelow z_n = b_p \waybelow z$
  and therefore $b_q \below c_r$ as $p, q \leq r \leq r + o$. As
  $c_r \below u$ we have that $z_n^k = b_q \below u$. As $k$ was
  arbitrary, this implies that $z_n = \sqsup_k z_n^k \below u$. By
  the same argument, as $n$ was arbitrary and $z = \sqsup_n z_n$, we
  may conclude that $z \below u$, thus establishing the claim.

  We now have that $x \waybelow z = \sqsup_n c_n$. Therefore, there
  weakly exists 
  $n \in \Nat$ such that $x \below c_n$.  Let 
  $c_n = \sqsup \lbrace b_i \mid i \in I \rbrace$ where $I \subseteq
  \lbrace 0, \dots, o+n \rbrace$ is a finite, non-empty set. For
  each $i \in I$ we can moreover find $b'_i \in \lbrace b_0, \dots,
  b_{n+o} \rbrace$ with $b_i \waybelow b_i' \waybelow z$. 

  By Lemma \ref{lemma:waybelow-sup} we have that $x \below \sqsup
  \lbrace b_i \mid i \in I \rbrace \waybelow \sqsup \lbrace b_i'
  \mid i \in I \rbrace \waybelow z$. Therefore $y = \sqsup \lbrace
  b_i' \mid i \in I \rbrace$ is our desired interpolant. This (only) shows
  weak existence of an interpolant, due to the weak existence of the
  number $n$ used in its construction. 
\end{proof}

\begin{cor}\label{cor:interpolation}
  Let $(B, \below)$ be a bounded complete predomain base
  for which $\waybelow$ is decidable. If $b_1, \dots, b_n \in B$ and
  $b_i \waybelow c$ for all $i = 1, \dots, n$ then there must exist an
  interpolant $b \in B$ such that $b_i \waybelow b \waybelow c$ for
  all $1 \leq i \leq n$.
\end{cor}
\begin{proof}
  By the previous lemma, we can find interpolants $\hat b_i$ for
  each $1 \leq i \leq n$ such that $b_i \waybelow \hat b_i \waybelow
  c$. By Lemma \ref{lemma:waybelow-sup} we have that $b =
  \sqsup\lbrace \hat b_i \mid 1 \leq i \leq n \rbrace$ satisfies $b
  \waybelow c$ and moreover $b_i \waybelow \hat b_i \below b$ so
  that $b_i \waybelow b$ for all $1 \leq i \leq n$.
\end{proof}

\noindent
We conclude the section with a technical lemma on swapping the order
of suprema that we will use later.

\begin{lemma} \label{lemma:sup-swap}
  Suppose that $P$ is a poset and $f: \Nat \times \Nat \to P$ is
  monotonic, i.e. $n \leq n'$ and $k \leq k'$ implies $f(n, k)
  \below f(n', k')$. Then
  \begin{enumerate}
    \item the sequence $(f(n, m))_n$ is monotonic for all $m \in
    \Nat$
    \item if $\sqsup_n f(n, m)$ exists for all $m \in \Nat$, then
      $(\sqsup_n f(n, m))_m$ is monotonic
    \item if both $\sqsup_n f(n, n)$ and $\sqsup_m \sqsup_n f(n, m)$
    both exist, they are equal.
  \end{enumerate}
\end{lemma}
\begin{proof}
  The first item is immediate by monotonicity of $f$. For the second
  item, fix $m \in \Nat$. We show that $\sqsup_n f(n, m) \below
  \sqsup_n f(n, m+1)$. This is immediate as $\sqsup_n f(n, m+1)$ is
  an upper bound of $f(n, m)$ for all $n \in \Nat$. 

  For the last item, suppose that $\sqsup_n f(n, n)$ and $\sqsup_m
  \sqsup_n f(n, m)$ both exist, in particular this entails that
  $\sqsup_n f(n, m)$ exists for all $m \in \Nat$. We first show that
  $\sqsup_n f(n, n)$ is an upper bound of all $\sqsup_n f(n, m)$ for
  all $n \in \Nat$. By monotonicity, we have that
  $\sqsup_n f(n, m) = \sqsup_{n \geq m} f(n, m) \below \sqsup_{n
  \geq m} f(n, n) = \sqsup_n f(n, n)$. To finish the proof, we need
  to show that $\sqsup_n f(n, n)$ is the least upper bound of
  $\sqsup_n f(n, m)$ for all $m \in \Nat$. So let $c$ be a
  competitor, i.e. $\sqsup_n f(n, m) \below c$ for all $m \in \Nat$.
  We show that $\sqsup_n f(n, n) \below c$. This follows once we
  establish that $f(n, n) \below c$ for all $n \in \Nat$ as
  $\sqsup_n f(n, n)$ is the least upper bound of all $f(n, n)$. So
  let $n \in \Nat$. But this is evident as $f(n, n) \below \sqsup_k
  f(k, n) \below c$ by assumption.
\end{proof}

\begin{cor} \label{cor:diagonal-sup}
  Let $P$ be a poset that has suprema of all increasing chains, and
  let $f: \Nat \times \Nat \to P$ be monotone. Then both $\sqsup_n f(n,
  n)$ and $\sqsup_m \sqsup_n f(n, m)$ exist and are equal.
\end{cor}

\section{Completion of Predomain Bases} \label{sec:completion}

We give a direct description of the rounded ideal completion of
\cite{Bauer:2009:CTC} with ideals being represented by chains.
The rounded ideal (or continuous) completion is distinguished from
the ideal completion by the definition of the order $\below$ on the
completion in terms of approximation $\waybelow$ on the underlying
predomain base, rather than in terms of its order.

\begin{defn}[Continuous Completion] \label{defn:cont-compl}
Let $(B, \below)$ be a predomain base. The \emph{continuous
completion} of $B$ is the set
$\hat B = \lbrace (b_n)_n \mid b_n \in B, b_n \below b_{n+1}
\mbox{ for all } n \in \Nat \rbrace$
of increasing sequences in $B$, with order relation defined by
\[ (b_n) \below (b'_n) \mbox{ iff } \forall b \in B\sep \forall n
\in \Nat \sep b \waybelow b_n \to \wex  m \in \Nat \sep b \waybelow b'_m \]
for increasing sequences $(b_n)_n$ and $(b'_n)_n$ in $B$.

The function $i: B \to \hat{B}$ that maps $b \in B$ to the constant
sequence $(b)_n$ is called the \emph{canonical embedding}, and in
the sequel we identify elements in $B$ with their canonical
embedding.
\end{defn}

\noindent
We show in Lemma \ref{4.8} that the canonical embedding indeed
preserves both $\below$ and $\waybelow$ which justifies our
terminology.
The above definition of the order $\below$ on the completion of a
predomain base showcases the first instance of our ``constructive
existence -- classical correctness'' approach in the classical
definition of the order relation on the completion above. In
particular, this implies that a realiser of $(b_n) \below (b'_n)_n$
carries no computational content. 
It is straightforward to see that the order relation $\below$
defined above is in fact a preorder. We omit the
straightforward proof of this fact.
\begin{lemma}
  The order relation $\below$ on the continuous completion $\hat B$
  of a predomain base $B$ is a preorder. 
\end{lemma}

\noindent
The preorder $\below$ on the continuous completion $\hat B$ of a
predomain base induces an equality relation on $\hat B$ where
$b = b'$ iff $b \below b'$ and $b' \below b$. In particular, this
gives $(\hat B, \below)$ the structure of a poset, where
arbitrary suprema, if they exist, are unique up to equality, i.e. $s
= \sqsup_i a_i$ and $s' = \sqsup_i a_i$ implies $s = s'$. Moreover,
suprema are extensional: if $c_i = d_i$ for all $i \in I$ and $s =
\sqsup_i c_i$, $t = \sqsup_i d_i$ then $s = t$.

It is an easy but very useful observation that every element
of the continuous completion is equal to the supremum of
the elements of (the canonical embeddings of) its representing sequence.

\begin{lemma}\label{lemma:own-sup}
  Let $B$ be a predomain base and $x = (x_n)_n \in \hat B$. Then $x
  = \sqsup_n x_n$.
\end{lemma}

\begin{proof}
  First, $x$ is an upper bound of all $x_n$. To see this, let $n \in
  \Nat$ and $a \in B$ with $a \waybelow x_n$. We have to show that
  there must exist some $k \in \Nat$ with $a \waybelow x_k$. But
  this clearly holds for $k = n$. To see that $x$ is the least upper
  bound of all $x_n$, consider a competing upper bound $c = (c_n)_n
  \in \hat B$ with $x_n \below c$ for all $n \in \Nat$. To see that $x
  \below c$, let $n \in \Nat$, $a \in B$ with $a \waybelow x_n$. As
  $x_n \below c$ there must exist $k \in \Nat$ such that $a
  \waybelow c_k$ which is precisely what we need to show for $x
  \below c$. 
\end{proof}

\noindent
We now show that the continuous completion $\hat B$ of a predomain base
$B$ has suprema of all increasing chains. Below, we write $M
\waybelow z$ if $x \waybelow z$ for all $x \in M$ and say that a
predomain base $B$ \emph{has weak interpolation} if there weakly exists $y$ such
that $M \waybelow y \waybelow z$ whenever $M$ is finite and $M
\waybelow z$. 

\begin{lemma} \label{lemma:sup-chains}
  Let $B = \lbrace b_n \mid n \in \Nat\rbrace$ be a predomain base
  and $(b_n^k)_k$ an approximating sequence of $b_n$. If
  $(c_n)_n$ is an increasing sequence in $\hat B$ and $c_n =
  (c_{n, m})_m$, then the following statements hold:
  \begin{enumerate}
    \item The set $D_k = \lbrace c_{n, m}^k \mid 0 \leq n, m \leq k
    \rbrace$ is consistent. If consistency is continuous, the same
    applies to the set $\hat D_k = \lbrace c_{n, k} \mid 0 \leq n \leq
    k \rbrace$.
    \item The sequence $(d_k)_k = (\sqsup D_k)_k$ is increasing. If
    consistency is continuous, the same applies to the sequence
    $(\hat d_k)_n = (\sqsup \hat D_k)_k$.
    \item If $B$ has weak interpolation, we have $c_n \below d$ for
    all $n \in \Nat$. If consistency on $B$ is moreover continuous,
    also $c_n \below \hat d$ for all $n \in \Nat$.
    \item If $c_n \below u$ for all $n \in \Nat$, then $d \below
    u$. If consistency on $B$ is continuous, then also $\hat d
    \below u$.

  \end{enumerate}
\end{lemma}
\begin{proof}
  For the first item, fix $k \in \Nat$ and let $0 \leq n, m \leq k$.
  As $(c_{n, m}^i)_i$ is an approximating sequence of $(c_{n, m})$
  we have $c_{n, m}^k \waybelow c_{n, m}$. As $c_n \below c_k$,
  there must exist $r = r(n, m)$ such that $c_{n, m}^k \waybelow
  c_{k, r}$. Let $s = \max \lbrace r(n, m) \mid 0 \leq n, m \leq k
  \rbrace$. Then $c_{n, m}^k \waybelow c_{k, r} \below c_{k, s}$ so
  that $c_{k, s}$ is an upper bound of $D_k$.

  Now assume that consistency is continuous. Then consistency of
  $\hat D_k$ follows if the sets $\hat D_{k, i} = \lbrace c_{n, k}^i
  \mid 0 \leq n \leq k \rbrace$ are consistent for all $i \in \Nat$.
  Let $i \in \Nat$ and $r = \max \lbrace i, k \rbrace$. By what we
  have just demonstrated, there must exist an upper bound $b$ of the
  set $D_r$. We show that $b$ is an upper bound of $\hat D_{k, i}$.
  This follows since for $c_{k, n}^i \in \hat D_{k, i}$ we have that
  $c_{k, n}^i \below c_{k, n}^r \in D_r$ and the fact that $b$ is an
  upper bound of $D_r$.

  The second item, monotonicity of $(d_k)_k$ and $(\hat d_k)_k$ is
  clear since both $D_k \subseteq D_{k+1}$ and $\hat D_k \subseteq
  \hat D_{k+1}$.

  For the third item, we begin by showing that $c_n \below d$. So
  fix $m \in \Nat$ and suppose that $x \waybelow c_{n, m}$, we show
  that there must exist $k \in \Nat$ such that $x \waybelow d_k$. As
  $c_{n, m} = \sqsup_k c_{n, m}^k$ there must exist $k \in \Nat$
  such that $x \waybelow c_{n, m}^k$. The same relation holds if we
  replace $k$ by $k' = \max \lbrace n, m, k \rbrace$ so we assume
  that $k \geq n, m$ without loss of generality. Then $x \waybelow
  c_{n, m}^k \below \sqsup \lbrace c_{n, m}^k \mid 0 \leq n, m \leq k
  \rbrace = d_k$ as required. Now assume that consistency is
  continuous. To see that $c_n \below \hat d$, fix $m \in \Nat$, $x
  \in B$ and assume that $x \waybelow c_{n, m}$. We show that there
  must exist $k \in \Nat$ such that $x \waybelow d_k$. This holds,
  for example, if $k = \max \lbrace n, m \rbrace$ for then $x
  \waybelow c_{n, m} \below c_{n, k} \below \sqsup \lbrace c_{n, k}
  \mid 0 \leq n \leq k \rbrace = \hat d_k$.

  For the last item, assume that $u = (u_i)_i \in \hat B$ and
  $c_n \below u$ for all $n \in \Nat$. We first show that $d \below
  u$. To see this, fix $k \in \Nat$, $x \in B$ and assume that $x
  \waybelow d_k$. We show that there must exist $s \in \Nat$ such
  that $x \waybelow u_s$. By assumption, we have $x \waybelow d_k =
  \sqsup \lbrace c_{n, m}^k \mid 0 \leq n, m \leq k \rbrace$. Fix $0
  \leq n, m \leq k$. Since $c_{n, m}^k \waybelow c_{n, m}$ and $c_n
  \below u$, there must exist $r = r(n, m) \in \Nat$ such that
  $c_{n, m}^k \waybelow u_r$. If $s = \max \lbrace r(n, m) \mid 0
  \leq n, m \leq k \rbrace$, we have $c_{n, m}^k \waybelow u_s$ for
  all $0 \leq n, m \leq k$. Hence, by Corollary
  \ref{cor:waybelow-sup-bounded-complete} we obtain
  $d_k = \sqsup \lbrace c_{n, m}^k \mid 0 \leq n, m \leq k \rbrace
  \waybelow u_s$ as required. Now suppose that consistency is
  continuous. To see that $\hat d \below u$, fix $k \in \Nat$ and $x
  \in B$ such that $x \waybelow \hat d_k = \sqsup \lbrace c_{n, k}
  \mid 0 \leq n \leq k \rbrace$. We show that there must exist $s
  \in \Nat$ such that (as above) $x \waybelow u_s$. 
  Fix $0 \leq n \leq k$. As $x \waybelow \hat d_k
  = \sqsup_{0 \leq k \leq n} \sqsup_i c_{n, k}^i = \sqsup_i
  \sqsup_{0 \leq n \leq k} c_{n, k}^i$ there must exist $i \in \Nat$
  such that $x \below \sqsup_{0 \leq n \leq k} c_{n, k}^i$. 
  For this $i \in \Nat$, we moreover have that $c_{n, k}^i
  \waybelow c_{n, k}$ so that there must exist $r = r(n, k, i)$ for
  which $c_{n, k}^i \waybelow u_r$ since $c_n \below u$. Hence for
  $s = \max \lbrace r(n, k, i) \mid 0 \leq n \leq k \rbrace$ we have
  that $c_{n, k}^i \waybelow u_s$ so that $\sqsup_{0 \leq k \leq n}
  c_{n, k}^i \waybelow u_s$ by Corollary
  \ref{cor:waybelow-sup-bounded-complete}, and finally
  $x \below \sqsup_{0 \leq k \leq n} c_{n, k}^i \waybelow u_s$ as
  desired.
\end{proof}

\noindent
The last lemma finally puts us into a position to show that the
completion of a predomain base is in fact complete.

\begin{cor} \label{cor:bounded-complete-complete}
Let $(B, \below)$ be a bounded complete predomain base. Then $(\hat
B, \below)$ has suprema of all increasing chains.
\end{cor}
\begin{proof}
  For bounded complete predomain bases, we have established the weak
  interpolation property in Corollary \ref{cor:interpolation}. The
  claim follows from Lemma \ref{lemma:sup-chains}.
\end{proof}

\noindent
We have the following extension theorem.

\begin{propn}
Suppose that $B$ and $C$ are predomain bases for which consistency
is continuous, and suppose that $C$ has weak
interpolation. Then every Scott continuous map $f: B \to \hat C$ has a
Scott continuous extension $\hat f: \hat B \to \hat C$.
\end{propn}

\begin{proof}
  Let $(x^i)_i$ be an approximating sequence for each element $x \in
  B$. 
  Define $\hat f( (b_n)_n) = (f(b_n))_n$ for a monotone sequence
  $(b_n)_n \in \hat B$. Then $\hat f((b_n)_n)$ is monotone as $f$ is
  monotone. We show that $\hat f$ is Scott continuous. For this we
  fix an increasing sequence $(b_n)_n \in \hat B$ where $b_n =
  (b_{n, m})_m$ and use Lemma \ref{lemma:sup-chains} to establish
  that
  \[
  A_k
   = [ \hat f(\sqsup_n b_n)]_k 
   = [ \hat f(\sqsup_{n \leq k} ( b_{n, k})_k)]_k
   = f(\sqsup_{n \leq k} b_{n, k})
  \]
  for the $k$-th element $[ \cdot ]_k$ of $\hat f(\sqsup_n b_n)$ using that
  consistency on $C$ is continuous. 
  For the $k$-th element of $\sqsup_n \hat f(b_n)$ we similarly
  obtain
  \[
     B_k  = [ \sqsup_n \hat f(b_n)]_k 
      = [ \sqsup_n (f ( (b_{n, k})_k) ]_k 
      = \sqsup_{n \leq k} f(b_{n, k})
  \]
  also using Lemma \ref{lemma:sup-chains} and continuity of
  consistencey on $B$. 
  For the claim, we need to establish $(A_k)_k = (B_k)_k$. To see
  that $(A_k)_k \below (B_k)_k$ fix $k \in \Nat$ and $x \in C$ such
  that $x \waybelow A_k = f(\sqsup_{n \leq k} b_{n, k})$. We show
  that there must exist $l \in \Nat$ such that $x \waybelow B_l =
  \sqsup_{n \leq l} f(b_{n, l})$. As $C$ has
  weak interpolation, there must exist $y \in C$ such that
  $x \waybelow y \waybelow f(\sqsup_{n \leq k} b_{n, k})$. Therefore
  \[ y \waybelow f(\sqsup_{n \leq k} \sqsup_{i \in \Nat} b_{n, k}^i)
  = f(\sqsup_i \sqsup_{n \leq k} b_{n, k}^i) = 
  \sqsup_i f(\sqsup_{n \leq k} b_{n, k}^i) \]
  using continuity of $f$. 
  Therefore there must exist $i \in \Nat$ such that
  $x \waybelow y \below f(\sqsup_{n \leq k} b_{n, k}^i)$. As $b_n
  \below b_k$ for $n \leq k$ and $b_{n, k}^i \waybelow b_{n, k}$
  there must exist $j(n)$ such that $b_{n, k}^i \waybelow b_{k,
  j(n)}$. Let $j = \max \lbrace j(n) \mid 0 \leq n \leq k\}$. Then $x
  \waybelow f(\sqsup_{n \leq k} b_{n, k}^i) \below f(b_{k, j})$. For
  $l = \max \lbrace j, k \rbrace$ we therefore obtain that
  $x \waybelow f(b_{k, j}) \below \sqsup_{n \leq l} f(b_{n, l}) =
  B_l$.
  For the reverse direction $(B_k)_k \below (A_k)_k$ fix $k \in
  \Nat$ and $x \in C$ such that $x \waybelow B_k = \sqsup_{n \leq k}
  f(b_{n, k})$. We show that there must exist $l \in \Nat$ such that
  $x \waybelow f(\sqsup_{n \leq l} b_{n, l})$. But this is evident
  for $l = k$ as $f(b_{n, k}) \below f(\sqsup_{n \leq k} b_{n, k})$
  by monotonicity of $f$, for all $0 \leq k \leq n$, whence $x
  \waybelow \sqsup_{n \leq k} f(b_{n, k}) \below f(\sqsup_{n \leq k}
  b_{n, k})$.
\end{proof} 

\noindent
We now show that bounded completeness (Definition
\ref{defn:predomain-bases}) transfers from a predomain
base to its completion.

\begin{lemma} \label{lemma:completion-bounded-complete}
  Let $(B, \below)$ be a bounded complete predomain base and suppose
  that consistency on $B$ is continuous. Let $I \subseteq \hat B$ be
  finite and consistent.
  \begin{enumerate}
  \item $I_n = \lbrace x_n \mid (x_n)_n \in I \rbrace$ is consistent
  for all $n \in \Nat$.
  \item $\sqsup I = (\sqsup I_n)_n$
  \end{enumerate}
  i.e. the completion of a bounded complete predomain base is
  bounded complete, and finite suprema of consistent sets are
  calculated pointwise.
\end{lemma}
\begin{proof}
  For the first item, let $b = (b_n)_n$ be an upper bound of $I$, and let
  $(x^i)_i$ be an approximating sequence of $x \in B$. As
  consistency on $B$ is continuous, it suffices to show that
  $\lbrace x_n^i \mid (x_n)_n \in I \rbrace$ is continuous for all
  $n, i \in \Nat$. So let $i, n \in \Nat$. For $x = (x_n)_n \in I$, 
  as $x_n^i \waybelow x_n$ and $(x_n)_n \below b$, there must exist
  $k(x)$ such that $x_n^i \waybelow b_{k(x)}$. Let $k = \max \lbrace
  k(x) \mid x \in I \rbrace$. Then $b_k$ is an upper bound of
  $\lbrace x_n^i \mid (x_n)_n \in I \rbrace$.

  For the second item, note that $\sqsup I_n \in B$ exists since
  $I_n$ is consistent and $(\sqsup I_n)_n$ is monotone as all $x \in
  I$ are monotone so that $(\sqsup I_n)_n \in \hat B$. We first show
  that $(\sqsup I_n)_n$ is an upper bound of all $x \in I$. So let
  $x \in I$, $n \in \Nat$ and assume that $z \in B$ with $z
  \waybelow x_n$. Then $z \waybelow \sqcup I_n$ whence $x \sqsubseteq (\sqcup I_n)_n$ by
  definition of $\below$ on $\hat B$. We now show that $(\sqsup
  I_n)_n$ is the least upper bound of $I$. So assume that $x \below
  b$ for all $x \in I$. We show that $(\sqsup I_n)_n \below b$. So
  let $n \in \Nat$, $z \in B$ with $z \waybelow \sqsup I_n$. We show
  that there must exist $k \in \Nat$ with $z \ll b_k$. As $\sqsup
  I_n = 
  \sqsup \lbrace \sqsup_i x_n^i \mid (x_n)_n \in I \rbrace =
  \sqsup_i \sqsup \lbrace x_n^i \mid (x_n)_n \in I \rbrace$ (using
  Lemma \ref{lemma:sup-chain}) and $z
  \waybelow \sqsup I_n$, there must exist $l \in \Nat$ such that $z
  \below \sqsup \{ x_n^k \mid (x_n)_n \in I \rbrace$. For $x = (x_n)_n
  \in I$ we moreover have that $x_n^k \waybelow x_n$, and since $x
  \below b$ there must exist $l(x) \in \Nat$ such that $x_n^i
  \waybelow b_{l(x)}$. Let $l = \max \lbrace l(x) \mid x \in I
  \rbrace$. Then $z \below \sqsup \lbrace x_n^k \mid (x_n)_n \in I
  \rbrace \waybelow b_l$ by Lemma \ref{lemma:waybelow-sup}.
\end{proof}

\noindent
We conclude the section with properties of the canonical embedding
that have already been reported in \cite{Bauer:2009:CTC}, albeit in
a slight different setting.

\begin{lemma}\label{4.8}
  Let $B$ be a bounded complete predomain base. Then the embedding
  $B \incl \hat B$ preserves and reflects $\below$ and preserves $\waybelow$.
\end{lemma}
\begin{proof}
  It is clear that the embedding preserves and reflects the order
  $\below$. To see that it preserves $\waybelow$, assume that $b, c
  \in B$ with $b \waybelow c$ in $B$. We show that $b \waybelow c$
  in $\hat B$ where we identify $b$ and $c$ with the constant
  sequences $(b)_n$ and $(c)_n$. 

  To see that $b \waybelow c$ in $\hat B$, let $(x_n)_n$ be a chain
  in $\hat B$ with $c \below \sqsup_n x_n$. Let $x_n = (x_{n, m})_m$
  and choose an approximating
  sequence $(a^n)_n$ for every element $a \in B$. As $\sqsup_n x_n =
  (\sqsup \lbrace x_{n, m}^k \mid 0 \leq n, m \leq k \rbrace)_k$ by Lemma
  \ref{lemma:sup-chains} and $b \waybelow c$, there must exist $k
  \in \Nat$ such that $b \below \sqsup \lbrace x_{n, m}^k \mid 0
  \leq k \leq n \rbrace$ in $B$ and hence also in $\hat B$. For $0 \leq n, m \leq k$ we have $x_{n,
  m}^k \below x_{n, m} \below x_n \below x_k$ and so $b \below x_k$
  as required.
\end{proof}

\section{The Scott topology} \label{sec:topology}

We investigate the Scott topology on the completion of a predomain
base. We show that every topologically continuous function is Scott
continuous, and that the Scott topologgy is generated by the
way-below relation. These results are also used later in the
function space construction. In particular, we employ a classical
definition of the Scott topology to align open sets with the
way-below relation.

\begin{defn}[Basic Notions]
  Let $(P, \below)$ be a poset. A subset $\SO \subseteq P$ is
  \emph{Scott-open} if it is an upper set ($\forall x, y \in P \sep
  x \below y \land x \in P \to y \in P$) and inaccessible by
  suprema of increasing chains (if $\sqsup_n x_n \in O$ for an
  increasing chain $(x_n)_n$ then there must exist $k
  \in \Nat$ such that $x_k \in \SO$).
\end{defn}

\begin{lemma} \label{lemma:scott-basis}
Let $B$ be a bounded complete predomain base.
Then $b\upup = \lbrace x \in \hat B \mid b \waybelow x
\rbrace$ is Scott-open for all $b \in \hat B$.
\end{lemma}
\begin{proof}
  To see that $b\upup$ is an upper set, assume that $x, y \in \hat
  B$ with $x \below y$ and let $x \in b\upup$, i.e. $b \waybelow x$.
  Then $b \waybelow x \below y$ so that $b \waybelow y$ by Lemma
  \ref{lemma:way-below-below}. To see that $b\upup$ is inacessible
  by suprema increasing chain increasing chainss, assume that $\sqsup_n x_n \in b\upup$, i.e.
  $b \waybelow \sqsup_n x_n$. By definition of $\waybelow$, there
  must exist $n \in \Nat$ such that $b \waybelow x_n$, hence $x_n
  \in b\upup$.
\end{proof}

\begin{lemma}
  Let $B$ be a bounded complete predomain base, $x \in \hat B$ and
  $\SO \subseteq 
  \hat B$ be Scott-open with $x \in \SO$. Then there must exist
  $b \in B \cap \SO$ such that $x \in b\upup \subseteq \SO$. 
\end{lemma}
\begin{proof}
  Choose an approximating sequence $(a^n)_n$ for every $a \in B$ and
  let $x = (x_n)_n$. Then $x = \sqsup_n x_n$ by Lemma
  \ref{lemma:own-sup} and $x_n = \sqsup_m x_n^m$ by assumption so
  that $x = \sqsup_n x_n = \sqsup_n \sqsup_m x_n^m = \sqsup_n
  \sqsup_m (\sqsup_{i \leq n} x_i^m )= \sqsup_n \sqsup_{i \leq n}
  x_i^n$ by Corollary \ref{cor:diagonal-sup}. As $\SO$ is
  Scott-open, there must exist $n \in \Nat$ such that $\sqsup_{i
  \leq n} x_i^n \in \SO$. For $i \leq n$, we have
  $x_i^n \waybelow x_i \below x_n$ so that $\sqsup_{i \leq n} x_i^n
  \waybelow x_n$ and therefore $\sqsup_{i \leq n} x_i^n \waybelow x$ by
  Lemma \ref{lemma:waybelow-sup}. As $B$ is bounded complete, we
  have $b = \sqsup_{i \leq n} x_i^n \in B$ so that $b \in B \cap
  \SO$. As $\SO$ is an upper set, it follows that $b\upup \subseteq
  \SO$. 
\end{proof}

\begin{remark}
  Taken together, the last two lemmas show that the Scott topology
  on $\hat B$ is generated by sets of the form $b\upup$ for $b \in
  B$.
\end{remark}

\noindent
The following lemma shows that Scott continuity, as defined in
Definition \ref{defn:predomain-bases} as preservation of suprema of
increasing chains, is also 
constructively equivalent to preservation of open sets under inverse
image.

\begin{lemma}
  Let $B$ and $C$ be bounded complete predomain bases and $f: \hat B
  \to \hat C$ be Scott-continous. Then $f\inv(\SO)$ is open whenever
  $\SO \subseteq \hat C$ is open.
\end{lemma}

\begin{proof}
  Let $\SO \subseteq \hat C$ be Scott-open. We show that
  $f\inv(\SO)$ is an upper set. So let $x \below y \in \hat B$ with
  $x \in f\inv(\SO)$, i.e. $f(x) \in \SO$. By monotonicity of $f$,
  we have $f(x) \below f(y)$ and as $\SO$ is an upper set, we have
  $f(y) \in \SO$, i.e. $y \in f\inv(\SO)$. We now show that $f\inv$
  is inaccessible by suprema of increasing chains. So let $\sqsup_n x_n \in
  f\inv(\SO)$. Then $\sqsup_n f(x_n) = f(\sqsup_n x_n) \in \SO$ and
  since $\SO$ is inaccessible by directed suprema, there must exist
  $n \in \Nat$ such that $f(x_n) \in \SO$ whence $x_n \in
  f\inv(\SO)$.
\end{proof}

\noindent
We need the following technical result before we can prove the
converse of the above lemma.

\begin{lemma} \label{lemma:approx-equal}
Let $B$ be a bounded complete predomain base with
decidable $\waybelow$ and choose an
approximating sequence $(b^i)_i$ for all $b \in B$. 
Furthermore, let $x = (x_n)_n \in \hat B$ and $y = (y_n)_n = (\sqsup_{i \leq
n} x_i^n)_n$. Then y is well defined (i.e. $\lbrace x_i^n \mid i
\leq n \rbrace$ is consistent) for all $n \in \Nat$ and $x = y$.
\end{lemma}

\begin{proof}
  It is evident that $\lbrace x_i^n \mid 0 \leq i \leq n \rbrace$ is
  consistent for all $n \in \Nat$, as for $0 \leq i \leq n$ we have
  that $x_i^n \below x_i \below x_n$, hence $x_n$ is the required
  upper bound.

  To see that $x \below y$ let $n \in \Nat$ and $b \in B$ with $b
  \waybelow x_n$. We show that there must exist $k \in \Nat$ with $b
  \waybelow y_k$. First, by interpolation, there must exist $c \in
  B$ such that $b \waybelow c \waybelow x_n$. As $(x_n^i)_i$ is an approximating sequence, there
  must exist $i \in \Nat$ such that $b \waybelow c \below x_n^i$,
  hence $b \waybelow x_n^i$. Let $k = \max
  \lbrace i, n \rbrace$. Then $b \below x_n^i \below x_n^k \below
  \sqsup_{i \leq k} x_n^k = y_k$.

  To see that $y \below x$, let $b \in \Nat$ and $b \waybelow y_n$.
  Put $k = n$. Then $b \waybelow y_n = \sqsup_{i \leq n} x_i^n
  \below \sqsup_{i \leq n} x_i \below x_n = x_k$.
\end{proof}
\begin{lemma}
  Let $f: \hat B \to \hat C$ be a function between bounded complete
  predomain bases where $\waybelow$ is decidable on $C$ 
  such that $f\inv(\SO)$ is open for all open $\SO
  \subseteq \hat C$. Then $f$ is Scott continuous.
\end{lemma}

\begin{proof}
  We first show that $f$ is monotone. So let $x \below y \in \hat
  B$. To see that $f(x) \below f(y)$ let $n \in \Nat$, $c \in C$ and
  suppose that $c \waybelow f(x)_n$. We show that there must exist
  $k$ such that $c \waybelow f(y)_k$. Since $c  \waybelow f(x)_n
  \below f(x)$ we have that $f(x) \in c\upup$ so that $x \in
  f\inv(c\upup)$ which is open and therefore an upper set. As $x
  \below y$, we have $y \in f\inv(c\upup)$ so that $c \waybelow
  f(y)$. If $(f(y)_n^i)_i$ is an approximating sequence of $f(y)_n$,
  Lemma \ref{lemma:approx-equal} gives $c \waybelow \sqsup_n \sqsup_{i
  \leq n} f(y)_i^n$. Therefore there must exist $k$ such that $c
  \below \sqsup_{i \leq k} f(y)_i^k$. We have, for $i \leq k$, that
  $f(y)_i^k \waybelow f(y)_i \below f(y)_k$ so that $\sqsup_{i \leq
  k} f(y)_i^k \waybelow f(y)_k$ by Lemma \ref{lemma:waybelow-sup} as
  required.

  Now suppose that $x = \sqsup_n x_n \in \hat B$, we show that $f(x)
  = \sqsup_n f(x_n)$. As $x_n \below x$, it is clear that $\sqsup_n
  f(x_n)  \below f(x)$. To see that $f(x) \below \sqsup_n f(x_n)$,
  let $(f(x_n)_m^i)_i$ be an approximating sequence of $f(x_n)_m$
  for all $n, m \in \Nat$. Then $(\sqsup_n f(x_n))_k = \sqsup_{0
  \leq n, m \leq k} f(x_n)_m^k$ by Lemma \ref{lemma:sup-chains} so
  that we need to show that $f(x) \below ( \sqsup_{0
    \leq n, m \leq k} f(x_n)_m^k)_k$. So let $c \in C$, $n \in \Nat$
    and suppose that $c \waybelow f(x)_n$. By Corollary
    \ref{cor:interpolation}, there must exist an interpolant  $c' \in C$ such that $c
    \waybelow c' \waybelow f(x)_n$. We show that there must
    exist $k \in \Nat$ such that $c \waybelow \sqsup_{n, m \leq k}
    f(x_n)_m^k$. 

    As $c' \waybelow f(x)_n \below f(x)$ we have that $x \in
    f\inv(c'\upup)$ which is open by assumption. As $x = \sqsup_n
    x_n$, there must exist $m$ such that $x_m \in f\inv(c'\upup)$,
    i.e. $c' \waybelow f(x_m)$. By Lemma \ref{lemma:approx-equal}, we
    have that $f(x)_m = \sqsup_n \sqsup_{i \leq n} f(x_m)_i^n$
    so that there must exist $n \in \Nat$ such that $c' \below
    \sqsup_{i \leq n} f(x_m)_i^n$. Let $k = \max \lbrace n, m
    \rbrace$. Then $c \waybelow c' \below \sqsup_{i \leq n} f(x_m)^i
    \below \sqsup_{i \leq n} f(x_m)^k \below \sqsup_{0 \leq n, m
    \leq k} f(x_m)^k$ as required.
\end{proof}

\section{Function Spaces}\label{sec:function-spaces}

Given two predomain bases $B$ and $C$, we construct a predomain base
$B \to C$ so that the continuous completion of $B \to C$ is the
space of continuous functions between the continuous completions of
$B$ and $C$. 

\begin{defn}
  Let $(B, \below)$ and $(C, \below)$ be predomain bases. A
  \emph{single step function} of type $B \to C$ is a pair $(b, c)
  \in B \times C$, written $b \step c$. A \emph{step function} is a finite set of single step functions, written
  $\sqsup_i b_i \step c_i$ such that
  \[ \Cons(\lbrace b_j \mid j \in J \rbrace) \to \Cons(\lbrace c_j
  \mid j \in J \rbrace ) \]
  for all (finite, non-empty) $J \subseteq I$. 
  Now suppose that $\waybelow$ on $B$ is decidable and $C$ is
  bounded complete. Then 
  a single step function $b \step
  c$ defines the function
  $b \step c: B \to C$ by $b \step c(x) = c$ if $b \waybelow x$, and
  $b \step c(x) = \bot$,
  otherwise. A step function $\sqsup_i b_i \step c_i$ defines the
  function $(\sqsup_i b_i \step c_i)(x) = \sqsup_i (b_i \step
  c_i(x))$. Step functions are ordered by $s \below t$ iff $s(x)
  \below t(x)$ for all $x \in B$. We write $B \to C$ for the set of
  step functions of type $B \to C$. 
\end{defn}
\noindent
Our first goal is to show that the collection of step functions
forms a predomain base, and we begin with collecting conditions that
ensure decidability of the ordering. The proof needs the following
condition that allows us to separate two subsets using the way below
relation.

\begin{defn} \label{defn:separated}
Let $(B, \below)$  be predomain base. Two subsets $A, D \subseteq B$
are \emph{separated} if there must exist $\omega \in B$
such that $a \waybelow \omega$ for all $a \in A$ and $d
\not\waybelow \omega$ for all $d \in D$.
\end{defn}

\noindent
Crucially, for the ordering on step functions to be decidable, we
need to require that separatedness is decidable. 

\begin{example}
  \begin{enumerate}
    \item Consider the predomain base $(B, =)$ from Example
    \ref{example:bases}, i.e. equality $=$ on $B$ is decidable. Then $A$ and $D$ are
    separated if $A$ is a singleton and $A \cap D = \emptyset$. As a
    consequence, separatedness between finite sets on $(B, =)$ is decidable. 
    \item For the predomain base $(B^\ast, \below_\pref)$ of finite
    sequences of $B$, we have that $A$ and $D$ are separated if
    there exists $a \in A$ such that every (other) $a' \in A$ is a
    prefix of $a$, and no $d \in D$ is a prefix of $a$. As a
    consequence, separatedness is decidable for finite subsets $A, D
    \subseteq B^\ast$.
  \end{enumerate}
\end{example}

\noindent
The situation with the predomain base of rational intervals is
slightly more complex, and warrants a separate lemma.

\begin{lemma}
 Two finite subsets $A, B \subseteq
  \IQ$ are separated   if and only if $\bigsqcup A$ exists and is not a
  singleton (i.e. $\upup \bigsqcup A$ is not empty) and $b
  \not\below \bigsqcup A$ for all $b \in B$. 
\end{lemma}
\begin{proof}
  Assume that $A = \lbrace [a_i,
  b_i] \mid i \in I \rbrace$ and $B = \lbrace [c_j, d_j] \mid j \in
  J \rbrace$ for finite sets $I$ and $J$. For the 'if' direction,
  assume moreover that $[a, b] = \bigsqcup A$ exists in $\IQ$ and
  $[c_j, d_j] \not\below \bigsqcup A$ for all $j \in J$. We show
  that $A$ and $B$ are separated by constructing a witness $\omega$
  of separation.
  Consider the sets $C$ and $D$ defined by
  \[ C = \lbrace c_i \mid a < c_i < b \rbrace \mbox{ and }
     D = \lbrace d_j \mid a < d_j < b \rbrace.
   \]
   We first consider the case where both $C$ and $D$ are non-empty,
   and let
   $c = \min C$, $d = \max D$, $\epsilon = \min \lbrace c,
   d \rbrace$ and $\delta = \max \lbrace c, d \rbrace$.
   It is immediate that 
   $a < \epsilon \leq \delta < b$
   from the definition of $C$ and $D$. We now claim that
   $\omega = [\frac{\epsilon + a}{2}, \frac{\delta + b}{2}]$
   witnesses separatednes of $A$ and $B$.

   First note that $\frac12 (\epsilon + a) \leq \frac12 (\delta +
   b)$ as $\epsilon \leq \delta$ and $a < b$ so that $\omega \in
   \IQ$. Moreover,
   $a < \frac{\epsilon + a}{2} \leq \frac{\delta + b}{2} < b$
   so that $[a, b] \waybelow \omega$ by Lemma
   \ref{lemma:IQ-waybelow} and as
   $[a_i, b_i] \below [a, b] \waybelow \omega$ we have that
   $[a_i, b_i] \waybelow \omega$ for all $i \in I$ using Lemma
   \ref{lemma:way-below-below}. To see that $[c_j, d_j] \not
   \waybelow \omega$ suppose that we have $j \in J$ with $[c_j, d_j]
   \waybelow \omega$. Applying Lemma \ref{lemma:IQ-waybelow} again,
   we obtain
   $c_j < \frac{\epsilon + a}{2} < \frac{\epsilon + \epsilon}{2} =
   \epsilon < b$
   and  similarly 
   $a < \delta = \frac{\delta + \delta}{2} < \frac{\delta + b}{2}
   < d_j$.
   As we have assumed that $[c_j, d_j] \not\below [a, b]$, we have
   that $a < c_j$ or $d_j < b$. In the first case, we obtain
   $c_j \in C$, and hence
   $c \leq c_j < \epsilon \leq c$, a contradiction. In the second
   case, we similarly have $d_j \in D$ and obtain a contradiction as
   $d \leq \delta < d_j \leq d$.

   We now consider the case where $C \neq \emptyset$ and $D =
   \emptyset$. Here, we define $\epsilon = \delta = \min C$ and let 
   $\omega = [ \frac{\epsilon + a}{2}, \frac{\delta + b}{2}]$ as
   before. Again, we have that $a < \delta = \epsilon < b$ and
   $a < \frac{\epsilon + a}{2} \leq \frac{\delta + b}{2} < b$ so
   that $\omega \in \IQ$ and $\omega \waybelow [a, b]$ whence
   $\omega \waybelow [a_i, b_i]$ for all $i \in I$.
   Again, if $[c_j, d_j] \waybelow \omega$ we obtain that 
   $c_j < \frac{\epsilon + a}{2} < \frac{\epsilon + \epsilon}{2} <
   b$ and $\delta = \frac{\delta + \delta}{2} < \frac{\delta + b}{2}
   < d_j$. As we have assumed that $[c_j, d_j] \not \below [a,
   b]$ we again have two cases: $a < c_j$ or $d_j < b$. 
   In the first case, we obtain $c_j \in C$ and as before we argue
   that then $c \leq c_j < \epsilon \leq c$ which is impossible.
   The second case is slightly different to the argument above, but
   we just need to observe that $d_j < b$ implies that $d_j \in D$,
   contradicting $D = \emptyset$.

   The case where $C = \emptyset$ and $D \neq \emptyset$ is entirely
   analogous and left to the reader. Finally, we consider the case
   where both $C$ and $D$ are empty.
   Here, we put $\omega = [\frac{a+b}{2}, \frac{a+b}{2}]$. Then $[a,
   b] \waybelow \omega$ as $a < b$, and by the same argument as
   before, $[a_i, b_i] \below [a, b] \waybelow \omega$ whence $[a_i,
   b_i] \waybelow \omega$ for all $i \in I$. To see that $[c_j, d_j]
   \not \waybelow \omega$ we again assume that $[c_j, d_j] \waybelow
   \omega$ and show that this is impossible. From the assumption
   $[c_j, d_j] \waybelow \omega$ we obtain that $c_j < \frac{a+b}{2}
   < d_j$, using Lemma \ref{lemma:IQ-waybelow} one more time. 
   Given that $c_j \notin C$, we furthermore obtain that $c_j \leq
   a$ or $c_j \geq b$. The latter case is impossible as then
   $\frac{a+b}{2} \leq \frac{b+b}{2} = b \leq c_j \leq d_j$ which
   cannot happen as we assumed that $[c_j, d_j] \waybelow \omega$.
   As a consequence, we have that $c_j \leq a$.

   For the same reason, given that $d_j \notin D$ we have that $d_j
   \leq a$ or $d_j \geq b$. Again the case $d_j \leq a$ is
   impossible (for a similar reason) so that $d_j \geq b$. But then
   $c_j \leq a \leq b \leq d_j$ so that $[c_j, d_i] \below \bigsqcup
   A$ which contradicts our assumption. This finishes the proof of
   the 
   'if'-implcation. 

   Conversely, suppose that $A$ and $B$ are separated and $\omega$
   is a witness of separatedness of $A$ and $B$. Then $\omega$ is an
   upper bound of $A$ so that $[a, b] = \bigsqcup A$ exists as $\IQ$
   is bounded complete. As $[a, b] \waybelow \omega$ by Lemma
   \ref{lemma:waybelow-sup} which implies that $a < b$ by Lemma
   \ref{lemma:IQ-waybelow}. We now claim that every $[c_j, d_j] \in
   D$ satisfies $[c_j, d_j] \not \below \bigsqcup A$. This
   follows, for if $[c_j, d_j] \below [a, b]$ we obtain $[c_j, d_j]
   \below [a, b] \waybelow \omega$ so that $[c_j, d_j] \waybelow
   \omega$ which is impossible as $\omega$ witnesses the separation
   of $A$ and $B$.
\end{proof}

\noindent
As the above characterisation deals with finite subset of $\IQ$ and
decidable properties only, the following Corollary is immediate.
\begin{cor}
  Separatedness is decidable for finite subsets $A, D \subseteq
  \IQ$.
\end{cor}


\noindent
We now use separatedness to characterise the order between step and
single step functions which takes us one step closer to the goal of
establishing decidability of the order on the set of step
functions.
\begin{thm}
Let $B$ be a predomain base for which $\waybelow$ is decidable, and
let $C$ be bounded complete. 
Furthermore, 
let $I$ be a finite set, $\lbrace \alpha \rbrace \cup \lbrace
\gamma_i \mid i \in I \rbrace \subseteq B$ and $\lbrace \beta
\rbrace \cup \lbrace \delta_i \mid i \in I \rbrace \subseteq C$ and 
and $\bigsqcup_{i \in I}  \gamma_i \step \delta_i$
be a step function. Then the following
are equivalent:
\begin{itemize}[label=$\triangleright$]
\item $\alpha \step \beta \below \bigsqcup_{i \in I} \gamma_i \step
\delta_i$
\item for all $x \in B$ with $\alpha \waybelow x$ we have that
$\beta \below \bigsqcup \lbrace \delta_i \mid i \in I \mbox{ and }
\gamma_i \waybelow x \rbrace$
\item
for any $I_0 \subseteq I$ such that $A=\{ \alpha \} \cup \{ \gamma_i
\mid i \in I_0 \}$ and $D=\{ \gamma_i \mid i \notin I_0 \}$ are
separated, we have $\beta \sqsubseteq \bigsqcup \{ \delta_i \mid  i \in I_0
\}$.  
\end{itemize}
\end{thm}
\begin{proof}
The equivalence of the first two items is immediate. So assume that
$\alpha \waybelow x$ implies that $\beta \below \bigsqcup \lbrace
\delta_i \mid \gamma_i \waybelow x \rbrace$ for all $x \in B$.
Moreover, let $I_0 \subseteq I$ be finite and assume that  
$A=\{ \alpha \} \cup \{ \gamma_i \mid i \in I_0
\}$ and $D=\{ \gamma_i \mid i \notin I_0 \}$ are separated. Note
that in this case, the supremum $\bigsqcup \lbrace \delta_i \mid i
\in I_0 \rbrace$ exists, as $\bigsqcup_{i \in I} \gamma_i \step
\delta_i$ is a step function, and $C$ is bounded complete. As
$\below$ is decidable, we may argue classically to show that $\beta
\below \bigsqcup \lbrace \delta_i \mid i \in I_0 \rbrace$. As
$A$ and $D$ are separated, we can (classically) find $w$ that
witnesses separatedness of $A$ and $D$ (as per Definition
\ref{defn:separated}). In particular, $\alpha \waybelow \omega$ so
that we have $\beta \below \bigsqcup \lbrace \delta_i \mid \gamma_i
\waybelow  \omega \rbrace =  \bigsqcup \lbrace \delta_i \mid i
\in I_0 \rbrace$. 

For the other direction, assume the last statement above and let $x
\in B$ with $\alpha \below x$. Put $I_0 = \lbrace i \in I \mid
\gamma_i \waybelow x \rbrace$. Then $A = \lbrace \alpha \rbrace \cup
\lbrace \gamma_i \mid i \in I_0 \rbrace$ and $D = \lbrace \gamma_i
\mid i \notin I_0 \rbrace$ are separated, and we obtain $\beta
\below \bigsqcup \lbrace \delta_i \mid i \in I_0 \rbrace = \bigsqcup
\lbrace \delta_i \mid \gamma_i \waybelow x \rbrace$ as required.
\end{proof}

\noindent
The above lemma is the key stepping stone to see that the ordering
on the predomain representing the function space is indeed
decidable.
\begin{cor} \label{cor:separated-order-decdiable}
Let $B$ and $C$ be predomain bases where $\waybelow$ on $B$ is
decidable, $C$ is bounded complete, and separatedness on $C$ is
decidable. Then $\below$ on $B \to C$ is decidable. 
\end{cor}
\begin{proof}
  Immediate by the previous lemma, as $\sqsup_{i \in I} \alpha_i
  \step \beta_i  \below \phi$ if and only if $\alpha_i \step \beta_i
  \below \phi$ for all $i \in I$.
\end{proof}

\noindent
Single step functions, and step functions themselves, are automatically
Scott continuous. 
\begin{lemma}
  Let $B$ and $C$ be predomain bases where $C$ is pointed,
  $\waybelow$ on $B$ is decidable, and $B$ has weak interpolation. 
  Then every single step function $b \step
  c$ is Scott continuous.
\end{lemma}

\begin{proof}
  Let $(x_n)_n$ be an increasing sequence in $B$ and $x \in B$ with
  $x = \sqsup_n x_n$. We show that $b \step c(\sqsup_n x_n) =
  \sqsup_n b \step c(x_n)$.  This is evident if $c = \bot$. So
  suppose $c \neq \bot$.
  By definition of single step functions (and the fact that $\waybelow$ on
  $B$ is decidable), we may distinguish two cases.

  \emph{Case 1.} $b \waybelow \sqsup_n x_n$ and $b \step c(\sqsup_n
  x_n) = c$. By weak interpolation
  on $B$, there must exists $y \in B$ such that $b \waybelow y
  \waybelow \sqsup_n x_n$, therefore there must exist $n \in \Nat$
  such that $b \waybelow y \below x_n$. Therefore $\sqsup_n b \step
  c(x_n) = c$. Note that the latter expression is a Harrop-formula,
  and therefore $\neg\neg$-stable.

  \emph{Case 2.} It is not the case that $b \waybelow \sqsup_n
  x_n$. We show that $\sqsup_n b \step c(x_n) = \bot$. This
  follows, if $b \step c(x_n) = \bot$ for all $n \in \Nat$. So pick
  $n \in \Nat$. We show that $\neg (b \waybelow x_n)$. Assume $b
  \waybelow x_n$. But then $b \waybelow \sqsup_n x_n$ whence $b
  \waybelow \sqsup_n x_n$ which entails $b \step c(\sqsup_n x_n) =
  \bot$, a contradiction.
\end{proof}

\noindent
The following lemma is standard in (classical) domain theory, e.g.
\cite[Proposition 4.0.2]{Abramsky:1994:DT}, and helps to construct
approximating sequences that in turn are required to show that
step functions form a predomain base.

\begin{lemma} \label{lemma:single-step-waybelow}
  Let $s$ be a step function. Then $b \step c \waybelow s$
  whenever $c \waybelow s(b)$. 
\end{lemma}
\begin{proof}
  Assume that $s \below \sqsup_n s_n$ where the $s_n$ are step
  functions.
  Then $s(b) \below \sqsup_n s_n(b)$. Since $c \waybelow s(b)$ we
  can find $n \in \Nat$ such that $c \below s_n(b)$. Let $x \in B$,
  we show that $b \step c(x)  \below s_n(x)$. In case $\neg(b
  \waybelow x)$ we have $b \step c(x) = \bot \below s_n(x)$. In case
  $b \waybelow x$ we have $b \step c(x) = c \below s_n(b) \below
  s_n(x)$ by monotonicity of $s_n$.
\end{proof}

\begin{lemma} \label{lemma:function-space-predomain-base}
  Let $B$ be a countable, bounded complete predomain base and
  suppose that $\waybelow$ and separatedness of finite sets on $B$
  are decidable, and let $C$ be
  bounded complete. Then the set $B \to C$ of step functions is a predomain
  base.
\end{lemma}

\begin{proof}
  We have to show that every step function $s \in B \to C$ has an
  approximating sequence. Let $B = \lbrace b_n \mid n \in \Nat
  \rbrace$. As $C$ is a predomain base, every $s(b_i) \in B$ has an
  approximating sequence $(e_i^j)_j$. Let $s_n = \sqsup_{0 \leq i
  \leq n} b_i \step e_i^n$. 

  We first show that $s_n$, for $n \in \Nat$, is indeed a step function.
  Let $J \subseteq \lbrace 0, \dots, n \rbrace$ be a non-empty
  subset such that $\lbrace b_j \mid j \in J \rbrace$ is consistent.
  We need to show that $\lbrace e_j^n \mid j \in J \rbrace$ is
  consistent. Pick $j \in J$. Then
  \[ e_j^n \waybelow s(b_j) \below s(\sqsup_{j \in J} b_j) \]
  so that $s(\sqsup_j b_j)$ is an upper bound of all $e_j^n$ for $j
  \in J$.

  It follows from Lemma \ref{lemma:single-step-waybelow} in
  combination with Corollary \ref{cor:waybelow-sup-bounded-complete} that $s_n
  \waybelow s$ for all $n \in \Nat$.

  We now show that $s \below \sqsup_n s_n$ for all $n \in \Nat$. Fix
  $n \in \Nat$ and let $x \in B$, we show that $s_n(x) \below s(x)$.
  This follows from
  \begin{align*}  s_n(x) 
    & = \sqsup \lbrace e_i^n \mid 0 \leq i \leq n, b_i \waybelow x \rbrace \\
    & \below  \sqsup \lbrace s(b_i) \mid 0 \leq i \leq n, b_i \waybelow x \rbrace \\
    & \below \sqsup \lbrace s(x) \mid 0 \leq i \leq n, b_i \waybelow x \rbrace \\
    & = s(x)
  \end{align*}
  so that $s$ is an upper bound of the $s_n$ in $B \to C$.

  We next establish that $s$ is in fact the least upper bound of the
  $s_n$. So let $t$ be a step function and suppose that $s_n
  \below t$ for all $n \in \Nat$. Let $x \in B$, we show that $s(x)
  \below t(x)$. As $s$ is a step function, we may assume that
  $s = \sqsup_{i \in I} b_i \step c_i$ for a finite set $I \subseteq
  \Nat$ and elements $c_i \in C$. Let $J = \lbrace i \in I \mid b_i
  \waybelow x \rbrace$ so that $s(x) 
  = \sqsup \lbrace c_j \mid j  \in J \rbrace$. By
  the interpolation Lemma (Corollary \ref{cor:interpolation}) we may find
  an interpolant $\hat x$ such that $b_j \waybelow \hat x \waybelow
  x$ for all $j \in J$. Then
  \[ b_i \waybelow x \mbox{ iff } b_i \waybelow \hat x \mbox{ for
  all }  i \in I \]
  so that $s(x) = s(\hat x)$. 
  As $B$ is countable, we may assume that $\hat x = b_i$ for some $i
  \in \Nat$. As $(e_i^j)_j$ is an approximating sequence for $s(x) =
  s(\hat x)$, it suffices to show that $e_i^n \below t(x)$. Now, for
  $n \geq i$, we have
  \begin{align*}
    e_i^n 
    & \below \sqsup_{0 \leq i \leq n} b_i \step e_i^n(x)
    && \mbox{(as $b_i = \hat x \waybelow x$ and $n \geq i$)} \\
    & = s_n(x) \below t(x)
  \end{align*}
  so that $e_i^n \below t(x)$ for all $n \in \Nat$ whence $s(x) =
  \sqsup_n e_i^n \below t(x)$ and finally $s \below t$ as $x$ was
  arbitrary. We note that this goal formula is stable whence
  applying the approximation lemma (that only guarantees classical
  existence of an interpolant) in fact proves the claim.

  We finally need to establish that $s_n \waybelow s$ for all $n \in
  \Nat$. But this is immediate from Lemma \ref{lemma:waybelow-sup}.
  The last requirement for a predomain base, decidability of
  ordering, has already been established in Corollary
  \ref{cor:separated-order-decdiable}.

\end{proof}

\begin{remark}
  It is in general not true that $c \waybelow c'$ implies that $b
  \step c \waybelow b \step c'$. If $b' \waybelow b$ in addition to
  $c \waybelow c'$ we have $b \step c \waybelow b' \step c'$. This
  fact cannot be exploited in the proof of the above Lemma as we
  have no way of approximating base elements from above.
\end{remark}

\noindent
We now show that every Scott continuous function arises as a
supremum of step functions.

\begin{lemma} \label{lemma:sup-simple}
  Let $B, C$ be predomain bases for which consistency is continuous,
  let $B$ be bounded complete with decidable $\waybelow$, and let
  $C$ be pointed. 
  Then, for every Scott continuous $f: \hat B \to \hat C$ there
  exists $s = (s_n)_n \in \widehat{B \to C}$ such that $s(x) =
  \sqsup_n  s_n(x) = f(x)$ for all $x \in \hat B$.
\end{lemma}

\begin{proof}
  Let $B = \lbrace b_0, b_1, \dots \rbrace$ and assume that
  $(b^i)_i$ is an approximating sequence for all $b \in B$. Suppose that $f:
  \hat B \to \hat C$ is Scott continuous. Let
  \[ s_n = \sqsup_{0 \leq i \leq n} b_i \step f(b_i)_n \]
  where we write $f(b_i)_n$ for the $n$-th element of the sequence
  $f(b_i) \in \hat C$. 
  Let $x = (x_n)_n \in \hat B$ be given. 
  
  We first show that $f(x) \below \sqsup_n s_n(x)$. 
  Let $y_n = \sqsup_{0 \leq i
  \leq n} x_i^n$ and $y = (y_n)_n$. Then $x = y$  by Lemma
  \ref{lemma:approx-equal}
  and it therefore
  suffices to show that
  \[ \sqsup_n f(y_n)\sqsubseteq \sqsup_n s_n(x)  \]
  as $f$ is extensional, i.e. $f(x) = f(y)$ and Scott continuous,
  i.e. $f(\sqsup_n y_n) = \sqsup_n f(y_n)$. We have the following
  calculations for the $k$-th elements of the respective sequences:
  \[ A_k = (\sqsup_n f(y_n))_k = \sqsup_{n \leq k} f(y_n)_k =
  \sqsup_{n \leq k} f(\sqsup_{i \leq n} x_i^n) 
  \]
  and
  \[ B_k = (\sqsup_n s_n(x))_k = \sqsup_{n \leq k} s_n(x_k) =
  s_k(x_k) = \sqsup \lbrace f(b_i)_k \mid 0 \leq i \leq k, b_i
  \waybelow x_k \rbrace 
  \]
  by Definition of $s_n$ and Lemma \ref{lemma:sup-chains}. We
  therefore need to show that $(A_k)_k \below (B_k)_k$. So let $k
  \in \Nat$ and $x \in C$ such that $x \waybelow A_k$. Choose $l$
  large enough so that $l \geq k$, $\lbrace x_0, \dots, x_k \rbrace
  \subseteq \lbrace b_0, \dots, b_l \rbrace$ and $\sqsup_{i \leq n}
  x_i^n \in \lbrace b_0, \dots, b_l \rbrace$ for all $n \leq k$.

  Now fix $n \leq k$. By our assumption on $l$, we have $j \leq l$
  such that $\sqsup_{i \leq n} x_i^n = b_j$. Moreover, for all $i
  \leq n$ we have $x_i^n \waybelow x_i \below x_l$ so that
  $b_j = \sqsup_{i \leq n} x_i^n \waybelow x_l$.
  Hence \[
    f(\sqsup_{i \leq n} x_i^n)_k = f(b_j)_k  \below 
    \sqsup \lbrace f(b_i)_k \mid 0 \leq i \leq l, b_i \waybelow x_l
    \rbrace = B_l. \]
  As $n$ was arbitrary, we therefore obtain
  $\sqsup_{n \leq k} f(\sqsup_{i \leq n} x_i^n) \below B_l$ and
  finally $x \waybelow B_l$ as required.

  We now show that $\sqsup_n  s_n(x) \below f(x)$. By
  Scott-continuity of $f$, we have $f(x) = \sqsup_n f(x_n)$ and the
  claim follows if $\sqsup_n s_n(x) \below \sqsup_n f(x_n)$. As
  above, we calculate for the $k$-th element of the respective
  sequences that
  \[ A_k = (\sqsup_n f(x_n))_k = \sqsup_{n \leq k} f(x_n)_k \]
  and
  \[ B_k = (\sqsup_n s_n(x))_k = \sqsup \lbrace f(b_i)_k \mid 0 \leq
  i \leq k, b_i
    \waybelow x \rbrace
  \]
  where we have used Lemma \ref{lemma:sup-chains} for the
  calculation of $A_k$, and the calcuation of $B_k$ is as above. To
  see that $(B_k)_k \below (A_k)_k$, fix $x \in C$, $k \in \Nat$ and
  assume that $x \waybelow B_k$. Let $N = \lbrace 0 \leq i \leq k
  \mid b_i \waybelow x_k \rbrace$ so that $B_k = \sqsup \lbrace
  f(b_i)_k \mid i \in N \rbrace$.

  We claim that $f(b_i) \below f(x_k)$ for every $i \in N$. This is
  immediate, since $i \in N$ implies that $b_i \waybelow x_k$,
  therefore $b_i \below x_k$ and $f(b_i) \waybelow f(x_k)$ by
  monotonicity of $f$. As a consequence, 
  $f(b_i)_k \below \sqsup_n f(b_i)_n = f(b_i) \below f(x_k)$ whence
  $\sqsup \lbrace f(b_i)_k \mid i \in N \rbrace \below f(x_k)$ where
  we view the left-hand side of the last equation as a constant
  sequence. As $x \waybelow B_k = \sqsup \lbrace f(b_i)_k \mid i \in
  N \rbrace$ there must exist $l \in \Nat$ such that $x \waybelow
  f(x_k)_l$. By monotonicity of $(f(x_k)_n)_n$, the same holds for
  $l$ replaced by $\max \lbrace l, k \rbrace$ so that we assume
  without loss of generality that $l \geq k$. In summary, we have
  obtained $x \waybelow f(x_k)_l \below \sqsup_{n \leq l} f(x_n)_l =
  A_l$ as required.
\end{proof}

\section{Real Numbers as Total Elements of the Interval Domain}
\label{sec:interval-domain}

We now consider an important example of predomain bases in more
detail, the predomain base of rational intervals that we have
already introduced in Example  \ref{ex:interval-predomain}.
Specifically, we introduce a constructive representation of the set
of real numbers as the total elements of the interval domain, and
give a characterisation of continuous functions on real numbers in
terms of Scott continuous functions on the interval domain.
Specifically, we compare total elements of the interval domain with
the (standard) constructive notion of Cauchy reals, and characterise
the total elements of the interval domain as \emph{Markov reals},
i.e. Cauchy reals where the modulus of convergence has been replaced
with a modulus of non-divergence.

In summary, we establish that if Markov's principle holds, Cauchy reals and Markov reals are
equivalent, and in turn equivalent to total reals in the sense of
domain theory.

\begin{defn}[Basic Notions] \label{defn:reals}
  We write $\IR$ for the continuous completion of $\IQ$, and write
  $\lo \alpha, \hi \alpha$ for the upper and lower endpoint of
  $\alpha \in \IQ$ as in Example \ref{ex:interval-predomain}, and
  identify $x \in \IR$ with the sequence $(x_n)_n$ so that $x =
  (x_n)_n = ([\lo x_n, \hi x_n])_n$ for $x \in \IR$. 
  The \emph{length} of $a \in \IQ$ is given by $\ell(a) = \hi a - \lo
  a$. An element $x \in \IR$ is a \emph{total real} or a
  simply a \emph{real}, if for all
  $k \in \Nat$ there must exist $n \in \Nat$ such that $\ell(x_n) \leq
  2^{-k}$.  We write $\Real$ for the set of total reals, and call
  a function $f: \IR \to \IR$ \emph{total} if $f(x) \in \Real$
  whenever $x \in \Real$.
  For $a, b \in \IQ$ we put
  \[ a + b = [\lo a + \lo b, \hi a + \hi b] \qquad -a = [- \hi a, - \lo a] \]
  and $|a| = a$ if $0 \leq \lo a$, $|a| = [0, \max \lbrace - \lo a,
  \hi a \rbrace]$ if $\lo a \leq 0 \leq \hi a$ and $|a| = -a$ if $\hi a \leq 0$.
  These operations are extended pointwise to
  elements $x, y \in \IR$, i.e. $x + y = (x_n + y_n)_n$, $-x =
  (-x_n)_n$, $|x| = (|x_n|)_n$ for
  all $n \in \Nat$. For $x = (x_n)_n \in \IR$ we put $0 \leq x$ if
  given $k \in \Nat$ there must exist $n \in \Nat$ such that
  $-2^{-k} \leq \lo x_n$ and $x \leq y$ if $0 \leq y - x$.
\end{defn}

\noindent
Given that the arithmetic operation above are defined on
equivalence classes of (elements of) the continuous completion of
$\IQ$, we need to show that they are well defined with respect to
the (defined) equality in the continuous completion (Definition
\ref{defn:cont-compl}). This is a standard verification. 

\begin{lemma}
  Let $x$, $y$, $x'$ and $y' \in \IR$ and suppose that $x = x'$ and
  $y = y'$. Then $x + y = x' + y'$, $-x = -x'$, $|x| = |x'|$ and $0
  \leq x$ if and only if $0 \leq x'$.
\end{lemma}

\begin{proof}
We write $x = (x_n)_n$, $x_n = [\lo x_n, \hi x_n]$ and similarly for $x'$, $y$ and $y'$. For an
element $z = [\lo z, \hi z] \in \IQ$ and $\epsilon \in \Rat$ with
$\epsilon \geq 0$ we write $z \pm \epsilon = [\lo
z - \epsilon, \hi z + \epsilon]$ as in Example
\ref{ex:interval-predomain}. We
begin with addition where it suffices to show that $x + y \below x'
+ y'$. So assume that $b \in \IQ$, $n \in \Nat$ and $b \waybelow (x
+ y)_n$. We show that there must exist $m$ such that $b \waybelow
(x' + y')_m$. Because $b \waybelow (x + y)_n$, there exists
$\epsilon > 0$ such that $b \waybelow (x + y)_n \pm \epsilon$.
Because $x \below x'$ there mus exist $m_x$ such that $x_n \pm
\frac{\epsilon}{2} \waybelow x'_{m_x}$. For the same reason, there
must exist $m_y$ such that $y_n \pm \frac{\epsilon}{2} \waybelow
y'_{m_y}$. For $m = \max \lbrace m_x, m_y \rbrace$ we therefore
obtain that
$b \waybelow (x + y)_n \pm \epsilon = x_n \pm
\frac{\epsilon}{2} + y_n \pm \frac{\epsilon}{2} \below x'_{m_x}
+ y'_{m_y} \below (x' + y')_m$ as required.

For unary minus, it similarly suffices to show that $-x \below -x'$.
If $b \waybelow -x_n$, Lemma \ref{lemma:IQ-waybelow} shows that $-b
\waybelow x_n$ whence there must exist $m$ such that $-b \waybelow
x'_m$ so that $b \waybelow -x'_m$ by applying Lemma
\ref{lemma:IQ-waybelow} again.

For the last claim, assume that $0 \leq x$, we show that $0 \leq
x'$. Let $k \geq 0$ be given. We show that there must exist $m \in \Nat$
such that $-2^{-k} \leq \lo x'_m$. 
As $0 \leq x$, there must exist $n$ such that $-2^{-(k+1)} \leq
\lo{x_{n}}$. As $x_{n} \pm 2^{-(k+1)} \waybelow x_{n}$, there must
exist $m$ such that $x_{n} \pm 2^{-(k+1)} \waybelow x'_{m}$
as $x = x'$ by assumption. In summary, we then obtain $-2^{-k} = 
-2^{-(k+1)} - 2^{-(k+1)} \leq  \lo{x_n} - 2^{-(k+1)}  \leq x_m$ as
required.

To show that $|x| = |x'|$ assume that $b \waybelow |x|_n$. As
ordering on $\Rat$ is decidable, we may distinguish the following
cases.

\emph{Case 1: $0 < \lo x_n$.} It is easy to see that in this case,
there must exist $m_0$ such that $0 < \lo x'_{m_0}$. Also, as $b
\waybelow |x_n| = x_n$, there
must exist $m_1$ such that $b \waybelow x'_{m_1}$. For $m = \max
\lbrace m_0, m_1 \rbrace$ we therefore have that $b \waybelow
x'_{m_1} \below x'_m$ and as $0 \leq \lo x'_{m_0} \leq \lo x'_m$
we have that $|x'_m| = x'_m$ so that $b \waybelow |x'_m|$ as
required.

\emph{Case 2:} $\lo x_n \leq 0 \leq \hi x_n$. In this case we have
$b \waybelow [0, \max \lbrace -\lo x_n, \hi x_n \rbrace]$ so that $\lo b <
0$ and $\hi b > \max \lbrace -\lo x_n, \hi x_n \rbrace$. We then
have that $[-\hi{b}, \hi{b}] \waybelow x_n$ and because $x = x'$
there must exist $m$ such that $[-\hi{b}, \hi{b}] \waybelow x'_m$.
Again the decidability of order on
$\Rat$ allows us to distinguish three subcases to relate this to
$|x'_m|$.

\emph{Subcase 2a: $\lo x'_m > 0$}. Then $\lo b < 0 < \lo x'_m \leq
\hi x'_m < \hi b$ as $-b < 0$ and $[-\hi{b}, \hi{b}] \waybelow x'_m$.
Hence $b = [\lo b, \hi b] \waybelow x'_m = |x'_m|$.

\emph{Subcase 2b: $\lo x'_m \leq 0 \leq \hi x'_m$.} As $[-\hi{b},
\hi{b} ] \waybelow x_m$ we have that $\max \lbrace -\lo x'_m,
\hi x'_m \rbrace < \hi{b}$. As $\lo{b} < 0$ as noted above, this
gives $\lo{b} < 0 \leq \max \lbrace -\lo x'_m,
\hi x'_m \rbrace < \hi{b}$ so that $b \waybelow [0,  \max \lbrace -\lo x'_m,
\hi x'_m \rbrace] = |x_m|$.

\emph{Subcase 2c: $\hi x'_m < 0$.} Similarly to Subcase 2a we
obtain that $-\hi{b} < \lo{x'_m} \leq \hi x'_m < 0 < -\lo{b}$ so
that
$\lo b < 0 < - \hi x'_m  \leq -\lo x'_m < \hi b$ and $[\lo b, \hi
b] \waybelow [-\hi x'_m, -\lo x'_m] = |x_m|$.

Our last case is analogous to the first:

\emph{Case 3: $\hi x_n < 0$.} As in the first case, it is easy to
see that there must exist $m_0$ such that $x'_{m_0} < 0$. Moreover,
$b \waybelow |x_n| = -x_n$ so that $-b \waybelow x_n$. As
consequence, there must exist $m_1$ such that $-b \waybelow
x'_{m_1}$. Let $m = \max \lbrace m_0, m_1 \rbrace$. Then $-b
\waybelow x'_m$ whence $b \waybelow -x'_m$ and $\hi x'_m  < 0$ so that
$|x'_m| = -x'_m$. Taken together this gives $b \waybelow |x'_m|$ as
required.

This argument shows that $|x| \below |x'|$ which is sufficient to establish
the penultimate claim.
\end{proof}

\noindent
The fact that the total reals coincide with the maximal elements of
$\IR$ is essentially a consequence that totality is formulated
negatively (and, like maximality,  has no computational content).
\begin{lemma} \label{lemma:total-max}
  The total reals are precisely the maximal elements of $\IR$.
\end{lemma}
\begin{proof}
  Let $x \in \Real$ be total and suppose that $x \below y$
  where $y \in \IR$. We show $y \below x$. So let $a \in \IR$, $n \in
  \Nat$ and suppose that $a \waybelow y_n$. We show that there must
  exist $k \in \Nat$ such that $a \waybelow x_k$. 
  By Lemma \ref{lemma:IQ-waybelow} there exists $\epsilon
  \in \IQ$, $\epsilon > 0$ so that $a \waybelow y_n \pm \epsilon$.
  As $x$ is an interval real, there must exist $k \in
  \Nat$ such that $\ell(x_k) \leq \frac \epsilon 2$. We show that $a
  \waybelow x_k$. Since $x_k \pm \frac \epsilon 2 \waybelow x_k$ and
  $x \below y$ there must exist $l \in \Nat$ such that $x_k \pm
  \frac \epsilon 2 \waybelow y_l$. We may assume that $l \geq n$. Using Lemma
  \ref{lemma:IQ-waybelow} this entails that $\lo x_k - \frac
  \epsilon 2 < \lo y_l \leq \hi y_l < \hi x_k + \frac \epsilon 2$ so
  that $\lo y_l - \hi x_k - \frac \epsilon 2 < 0$. Using this
  estimate, we obtain
  \[ \lo a < \lo y_n -  \epsilon \leq \lo y_n - (\hi x_k - \lo x_k +
  \frac \epsilon 2) \leq \lo x_k + \lo y_l - \hi x_k - \frac \epsilon 2 \leq
  \lo x_k. \]
  One analogously establishes that $\hi x_k < \hi a$ so that $a
  \waybelow x_k$ as desired.

  Now let $x$ be maximal. We show that $\forall k \in \Nat \sep \wex
  n \in \Nat \sep \ell(x_n) \leq 2^{-k}$. So assume that $k \in \Nat$
  and $\forall n \in \Nat\sep \neg \ell(x_n) \leq 2^{-k}$. We establish
  a contradiction. 

  Define $y_n = [\lo x_n + 2^{-n-1}, \hi x_n - 2^{-n-1}]$. Then $y =
  (y_n)_n \in \IR$ and $x \below y$. As $x$ is maximal, we have $y
  \below x$. Since $y_k \pm 2^{-k-1} \waybelow y_k$ and $y \below
  x$, there must exist $m \in \Nat$ such that $y_k \pm 2^{-k-1}
  \waybelow x_m$. We may assume that $m \geq k$. Then $\ell(y_k) +
  2^{-k} = \ell(y_k \pm 2^{-k-1}) < \ell(x_m)$. But then
  $\ell(x_k) = \ell(y_k) + 2^{-k} < \ell(x_m) \leq \ell(x_n)$, the desired
  contradiction. 
\end{proof}

\noindent
We characterise total reals as Cauchy sequences where we
reformulate the notion of Cauchyness to account for the lack of
information on convergence speed. 

\begin{defn}
A \emph{classical null sequence} is a decreasing sequence $(q_n)_n$
in $\Rat_{\geq 0} \cup \lbrace \infty \rbrace$ such that
$\forall \epsilon > 0 \sep \wex m
\in \Nat \sep \forall n \geq m \sep q_n \leq \epsilon$. We write
$(q_n)_n \fromabove 0$ if $(q_n)_n$ is a classical null-sequence.

A \emph{modulus of non-divergence} of a rational sequence $(q_n)_n$
is a classical null-sequence $(c_n)_n$ such that $\forall N \in \Nat
\sep \forall n, m \geq N \sep |q_n - q_m| \leq c_N$.

Finally, \emph{modulus of convergence} of a rational sequence $(q_n)_n$ is
a non-decreasing function $M: \Nat \to \Nat$ such that 
$\forall n, m \geq M(k) \sep |q_n - q_m| \leq 2^{-k}$.

A rational sequence $(q_n)_n$ is \emph{Cauchy} if it has a modulus
of convergence, and \emph{Markov} if it has a modulus of
non-divergence. A \emph{Cauchy real} (\emph{Markov real}) is a
rational sequence together with a modulus of convergence (modulus of
non-divergence).
\end{defn}

\noindent
We allow classical null sequences to begin with $\infty$ to account
for partially defined approximations of functions later. The
following is simple observation, but requires a rather technical
proof to allow us to convert between moduli of (dis-) continuity.

\begin{lemma} \label{lemma:cauchy-markov}
Every Cauchy sequence is Markov.
\end{lemma}

\noindent
The proof of this lemma needs the following auxiliary statement that
helps to convert the modulus into a (classical) null sequence.

\begin{lemma} \label{lemma:modulus}
Let $M: \Nat \to \Nat$. Then there exists a non-decreasing function $W: \Nat \to
\Nat$ such that $W(M(0)) = 0$, $(2^{-W(n)})_n$ is a classical null sequence and $M
\circ W(n) \leq n$ for all $n \geq M(0)$. 
\end{lemma}
\begin{proof}
We put $W(n) = 0$ for $n \leq M(0)$ and use induction for $n >
M(0)$, i.e.
\[
  W(n) = \begin{cases}
    W(n-1) & \mbox{ if } M(W(n-1) + 1) > n \\
    W(n-1) + 1 & \mbox{ if } M(W(n-1) + 1) \leq n
  \end{cases}
\]
for all $n > M(0)$.  It is clear that $W(M(0)) = 0$. 
We first claim that $W$ is non-decreasing, but this is evident from
the definition of $W$. Next, we show by induction on  $n$ that $M
\circ W(n) \leq n$ for all $n \geq M(0)$. For $n = M(0)$ this
follows directly from the definition, as then $M \circ W(n) = M(0)
\leq M(0) = n$. Now let $n > M(0)$. We distinguish two cases.
If $M(W(n-1) + 1) > n$, then $M \circ W(n) = M \circ W(n-1) \leq n-1
\leq n$. If on the other hand, $M(W(n-1) + 1) \leq n$, we obtain $M
\circ W(n) = M(W(n-1) + 1) \leq n$ by assumption.

We now show that $W$ is progressive, that is, $\forall n \geq
M(0)\sep \exists m \geq n+1\sep W(m) > W(n)$. Fix $n \geq M(0)$ and
put $m = n+k$ where $k = \max \lbrace 1, M(W(n) + 1) - M(W(n))
\rbrace$. As $W$ is monotone, it suffices to show that $W(n) =
W(n+1) = \dots = W(n+k)$ is contradictory as the latter chain of
equivalences is decidable. 

So assume that  $W(n) =
W(n+1) = \dots = W(n+k)$. In particular, as $W(n+k) = W(n+k-1)$, we
obtain $M(W(n+k-1) + 1) > n+k$. Since $W(n+k-1) = W(n)$ this gives
$M(W(n) + 1) > n+k$. Substituting the definition of $k$, this gives
$M(W(n) + 1) > n + M(W(n) + 1) - M(W(n))$, i.e. $M \circ W(n) > n$,
contradicting our earlier statement. 

We now show that $W(n)$ is not bounded, i.e. $\forall k \in \Nat
\sep \exists n \in \Nat \sep W(n) \geq k$. This follows by induction
on $k$, and for $k = 0$ we put $n = 0$. For the inductive step,
assume that there exists $n' \in \Nat$ such that $W(n') \geq k-1$.
By the previous claim, there exists $n \geq n'+1$ such that $W(n) >
W(n') \geq k-1$ so that $W(n) \geq k$.

This now shows that $(2^{-W(n)})_n$ is a classical null-sequence, for
any $\epsilon \geq 0$ there exists $k \geq \frac{1}{\epsilon} + 1$ and in
turn $n \in \Nat$ such that $W(n) \geq k$ whence $2^{-W(n)} \leq
2^{-k} \leq \frac 1k \leq \epsilon$ as $k \geq 1$. 
\end{proof}

\noindent
We now give the proof of Lemma \ref{lemma:cauchy-markov}.

\begin{proof}
Let $(q_n)_n$ be a Cauchy sequence with modulus $M$, i.e. $\forall k
\in \Nat \sep \forall n, m \geq M(k) , \hspace{1.5mm} |q_n - q_m| \leq 2^{-k}$.
By Lemma \ref{lemma:modulus} there exists a non-decreasing function $W: \Nat \to \Nat$ such
that $M \circ W(n) \leq n$ for all $n \geq M(0)$ and $(2^{-W(n)})_n$
is a classical null-sequence.

We define a sequence $(c_n)_n$ of non-negative rationals by
$c_n = 1 + \max \lbrace |q_k - q_l| \mid n \leq k, l \leq M(0)
\rbrace$ if $n < M(0)$ and $c_n = 2^{-W(n)}$ if $n \geq M(0)$. 

We now show that if $m, n \geq N$, then $|q_n - q_m| \leq c_N$.
To see this, we distinguish the cases $N < M(0)$ and $N \geq M(0)$.
First suppose that $N < M(0)$. If both $n, m \leq M(0)$ this follows
by definition of $c_N$. If both $n, m \geq M(0)$ then $|q_n - q_m|
\leq 2^0 = 1 \leq c_N$. Now consider without loss of generality that
$n < M(0)$ and $m \geq M(0)$. Then $|q_n - q_m| \leq |q_n -
q_{M(0)}| + |q_{M(0)} - q_m| \leq \max \lbrace |q_k - q_l | : N \leq
k, l \leq M(0) \rbrace + 1 = c_N$.

Now suppose that $N \geq M(0)$ and let $m, n \geq N$. Then $M \circ
W(N) \leq N$ so that $m, n \geq M(W(n))$ whence $|q_n - q_m| \leq
2^{-W(N)}= c_N$. It remains to show that $(c_n)_n$ is a classical
null-sequence which is however immediate from the fact that that
$2^{-W(n)}$ is a classical null sequence.
\end{proof}

\noindent
In fact, it turns out that the equivalence of Cauchy and Markov
reals characterises Markov's principle itself.

\begin{lemma}
  If Markov's Principle holds, every Markov sequence is Cauchy.
\end{lemma}
\begin{proof}
  Let $(q_n)_n$ be a Markov sequence with (Markov-) moduls
  $(c_n)_n$. As $c_n$ is a classical null-sequence, we obtain that
  $\forall k \in \Nat \sep \wex n \in \Nat \sep c_n \leq 2^{-k}$. By
  Markov's principle, we may replace the weak existential quantifier
  by a strong one, and number choice yields the Cauchy modulus.
\end{proof}

\begin{lemma} If every Markov Sequence is Cauchy, then Markov's
principle holds.
\end{lemma}
\begin{proof}
  Suppose that every Markov sequence is Cauchy, and let $P(n)$ be a
  decidable predicate on natural numbers. Suppose that $\wex n\sep
  P(n)$. We show that $\exists n\sep P(n)$. Define a sequence
  $(q_n)_n$ by $q_n = 0$ if $\forall k \leq n \sep \neg P(n)$, and
  $q_n = 1$, otherwise. We claim that $(q_n)_n$ is a Markov
  sequence. Put $c_n = 1-q_n$. Then clearly $|q_n - q_m| \leq c_N$
  for $n, m \geq N$ and $N \in \Nat$. If $c_N = 1$ then clearly $|q_n - q_m| \leq 1$. If,
  on the other hand, $c_N = 0$ we have $q_N = 1$ and therefore $q_n
  = q_m = 1$ as $n, m \geq N$. It remains to see that $(c_n)
  \fromabove 0$, i.e. $(c_n)_n$ is a classical null-sequence. It is
  clear that $c_n$ is decreasing. Now let $\epsilon > 0$. We show
  that $\wex n\sep c_n \leq \epsilon$ by showing that $\wex n\sep
  c_n = 0$. But we have $\wex n\sep P(n)$ hence $\wex n \sep q_n =
  1$ whence $\wex n \sep c_n = 0$.
\end{proof}

\noindent
We are now in a position to relate Markov reals to the total reals,
i.e. the maximal elements in the interval domain.  

\begin{lemma}
  Every Markov real is a total real, and vice versa. More precisely,
  the following hold:
  \begin{enumerate}
  \item 
  Every total real $x = (x_n)_n$ defines a Markov sequence $m(x) = (\frac
  12(\hi x_n + \lo x_n))$ with Markov modulus $c_n = \hi x_n - \lo
  x_n$. 
  \item 
  Every Markov-sequence $q = (q_n)_n$ with modulus
  $(c_n)_n$ defines a total real $t(q) = (x_n)_n = (\sqsup_{0 \leq i \leq n} q_i
  \pm c_i)_n$. 
  \item 
  If $x$ is a total real, then $x = t(m(x))$, that is, the above
  constructions perserve equality. 
  \end{enumerate}
\end{lemma}
\begin{proof}
  First suppose that $x = (x_n)_n = ([\lo x_n, \hi x_n])_n$ is a
  total real and let $q_n = \frac 12 (\hi x_n + \lo x_n)$. By
  definition, $c_n = \hi x_n - \lo x_n$ is a classical
  null-sequence, so that we just need to establish that $|q_n - q_m|
  \leq c_N$ whenever $n, m \geq N$. Let $k = \min \lbrace n, m
  \rbrace$ and $l = \max \lbrace n, m \rbrace$. Then $k, l \geq N$ and we have that 
  $q_n - q_m = \frac12 (\hi x_n + \lo x_n - \hi x_m - \lo x_m) \leq
  \frac 12 (\hi x_k + \lo x_l - \hi x_l - \lo x_k) = \frac 12
  (c_k + c_l) \leq c_N$ by monotonicity of $x$.  One analogously
  establishes that $-c_N \leq q_n - q_m$ whence $|q_n - q_m| \leq
  c_N$.

  For the converse, we just need to establish that $\lbrace q_i \pm
  c_i \mid 0 \leq i \leq n \rbrace$ is consistent, given a Markov
  sequence $(q_n)_n$ with modulus $(c_n)_n$ which is evident as for
  example $c_n \below q_i \pm c_i$ for all $0 \leq i \leq n$. 
  The last claim is evident by construction.
\end{proof}

\noindent
If we were to define the natural equality relation on Markov reals,
we would also be able to establish that the above constructions also
preserve equality of Markov reals, but this is not needed for what
follows.

\section{Example: Computation of Square Roots} \label{sec:sqrt}

Having established the interval domain as a way of reperesenting real
numbers, we now exemplify the claim that we made in the introduction by means
of an example: computation with constructive domain theoretic reals
gives actual-case error bounds that are much tighter than the worst
case error bounds of Cauchy reals.

Throughout the section, we fix a positive rational number $q > 0$
and
demonstrate how to compute a total real $s = ([\lo s_n, \hi s_n])_n$  such that
$\lo s_n^2 \leq q \leq \hi s_n^2$ for all $n \in \Nat$, i.e.
$s$ represents the square root of $q$.

Our definition is based on Newton iteration, specifically
\[ \hi s_{-1} = 1 \qquad \hi s_n = \frac 12 \left( \hi s_{n-1} +
\frac{q}{\hi s_{n-1}} \right)
\qquad \lo s_n = q / \hi s_n \]
for all $n \in \Nat$. We show that $s = ([\lo s_n, \hi s_n])_n$ is
indeed a total real, and represents the square root of $q$. This is
a sequence of lemmas involving standard estimates, the proofs of
which we elide. In particular:

\begin{lemma}\label{lemma:sqrt-props}
Let $s$ be given as above. 
  \begin{enumerate}
    \item Both $\lo s_n > 0$ and $\hi s_n > 0$ for all $n \in \Nat$.
    \item $(\hi s_n)_n$ is decreasing with $\hi s_n^2 \geq q$
    for all $n \in \Nat$.
    \item $(\lo s_n)_n$ is increasing with $\lo s_n^2 \leq q$
    for all $n \in \Nat$.
    \item $\lo s_n \leq \hi s_n$ for all $n \in \Nat$.
  \end{enumerate}
\end{lemma}

\noindent
We now need to demonstrate that $s$ is indeed total, i.e. the
distance $\hi s_n - \lo s_n$ can be made arbitrarily small. 
This is an immediate consequence of the fact that $(\lo s_n)$ is
increasing that we use in the following lemma.
\begin{lemma}
We have that $\hi s_n  - \lo s_n  \leq \frac12 (\hi s_{n-1}  -
\lo s_{n-1}) $ for all $n > 0$.
\end{lemma}
\begin{proof}
Immediate, as 
\[ \hi s_n  - \lo s_n  \leq \frac12 \left( \hi s_{n-1}  +
\frac{q}{\hi s_{n-1}} \right) - \lo s_{n-1}  \leq \frac12 \left( \hi s_{n-1}  -
\lo s_{n-1}  \right) \]
using that $\lo s_n$ is increasing.
\end{proof}

\noindent
The following is now immediate:
\begin{cor}
  With $\lo s_n$ and $\hi s_n$ defined as above, $s = (\lo s_n,
  \hi s_n)_n$ is a total real with $s^2 = q$.
\end{cor}

\noindent
We have the following comparision to Cauchy reals.

\begin{comparison} 
If $q = 2$, in particular we have that $|\hi s_n - \hi s_m| \leq
\frac{1}{6 \cdot 2^n}$ for all $m \geq n$. In other words, computing
the $n$-th iterate of the square root of two, the attained precision
is $\frac{1}{6 \cdot 2^n}$. This is precisely the same modulus of
convergence that was obtained for the very same example in \cite{Schwichtenberg:2016:CAW}. 
The following table compares this to the interval width
obtained from a domain theoretic approach where we have used a
simple (hand extracted) Haskell program to obtain the data reported.

\medskip
\begin{tabular}{|c|c|c|} \hline
Iterations & Interval Width & Modulus Precision \\ \hline \hline
1 & $4.9 \times 10^{-3}$ & $8.3 \times 10^{-2}$ \\ \hline
2 & $4.2 \times 10^{-6}$ & $4.2 \times 10^{-2}$ \\ \hline
3 & $3.2 \times 10^{-12}$ & $2.1 \times 10^{-2}$ \\ \hline
4 & $1.8 \times 10^{-24}$ & $1.0 \times 10^{-2}$ \\ \hline
5 & $5.7 \times 10^{-49}$ & $5.2 \times 10^{-3}$ \\ \hline
\end{tabular}

\medskip\noindent
While we cannot draw any valid conclusions from this very small
experiment, we nonetheless note that the difference in precision is
staggering, and warrant further investigation.
\end{comparison}

\section{Scott Topology vs Euclidean Topology}
\label{sec:scott-euclidean}

In Section \ref{sec:interval-domain}, we have investigated total reals
\emph{individually} by relating them to Markov and  Cauchy reals. We
continue our investigation of the real line induced by the total
elements of the interval domain by also considering the topology on
the real line, and linking that with the Scott topology of the
interval domain. 

As with convergence speed, we adopt a classical definition of
properties, in particular that of openness. 

\begin{defn}[Basic Notions]
  A subset $O \subseteq \Real$ is \emph{open}, if for all $x \in O$
  there must exist $k \in \Nat$ such that $\forall y \in \Real \sep
  |x - y| \leq 2^{-k} \to y \in O$.
\end{defn}

\noindent
Our first result shows that the topology defined above is the
subspace topology on the set of total reals, induced by the Scott
topology. That is, every open set arises as the intersection of a
Scott open set with the total reals, and every such intersection is
itself open.

\begin{lemma}
  Let $\SO \subseteq \IR$ be Scott-open. Then $\SO \cap \Real$ is
  open.
\end{lemma}
\begin{proof}
  Let $x \in \SO \cap \Real$. As $x \in \SO$, there must exist
  $\alpha \in \IQ$ so that $\SO \ni \alpha \waybelow x$ by Lemma
  \ref{lemma:scott-basis}. 
  As $x = \sqsup_n x_n$ by Lemma
  \ref{lemma:own-sup}, there must exist $n \in \Nat$ such that $\alpha
  \below x_n$ and hence $\alpha \waybelow x_n$. By Lemma
  \ref{lemma:IQ-waybelow} there exists $\epsilon \in \Rat_{>0}$ such
  that $\alpha \waybelow x_n \pm \epsilon$. We now claim that all $y
  \in \SO$ whenever $y \in \Real$ with $|x-y| \leq \epsilon$. So
  pick $y \in \Real$ and assume that $|y - x| \leq \epsilon$. Then
  $\alpha \waybelow x_n \pm \epsilon \below x\pm \epsilon \below y$
  by Lemma \ref{lemma:abs-below}. Hence $\alpha \waybelow y$ so that
  $y \in \SO$.
\end{proof}

\begin{lemma} Let $O \subseteq \Real$ be open. Then
\[ \SO = \lbrace \alpha \in \IR \mid \wex [p, q] \in \IQ, [p, q]
\below \alpha \mbox{ and } [p, q] \subseteq O \rbrace \]
is Scott-open and satisfies $\SO \cap \Real = O$.
\end{lemma}
\begin{proof}
  We show that $\SO$ is Scott-open. First, it is immediate that
  $\SO$ is an upper set. To see that $\SO$ is inacessible by
  directed suprema, let $(\alpha_n)_n$ be an increasing sequence in
  $\IR$ with $\sqsup_n \alpha_n \in \SO$. This gives $[p, q] \in
  \IQ$ with $[p, q] \below \alpha$ and $[p, q] \subseteq O$. As $O$
  is open, there must exist $\epsilon > 0$ and $\epsilon \in \Rat$ with
  $[p, q] \pm \epsilon \subseteq O$. To see that there must exist $m
  \in \Nat$ with $\alpha_m \in \SO$, we show that there must exist
  $m \in \Nat$ with $[p, q] \pm \epsilon \below \alpha_m$. Because
  $[p, q] \pm \epsilon \waybelow [p, q] \below \sqsup_n \alpha_n$
  this $m$ must exist, i.e. $m$ must exist with $[p, q] \pm \epsilon
  \below \alpha_m$.

  It remains to be seen that $\SO \cap \Real = O$. Let $x \in \SO
  \cap \Real$. Because $x \in \SO$ we have $[p, q] \in \IQ$ with
  $[p, q] \below x$ and $[p, q] \subseteq \SO$. As $O$ is open, this
  entails that there must exist $\epsilon > 0$ and $\epsilon \in
  \Rat$ such that $[p, q] \pm \epsilon \subseteq O$. We show that $p
  - \epsilon \leq x \leq q + \epsilon$ which implies $x \in O$. This
    follows from $[p, q] \pm \epsilon \waybelow [p, q] \below x$. It
    is clear that $O \subseteq \SO \cap \Real$.
\end{proof}

\noindent
We state and prove the following technical lemma to help deal with
the Euclidean topology.
\begin{lemma} \label{lemma:interval-arith}
  Let $a, b \in \IQ$. Then $a \leq b$ iff $-b \leq -a$ and $|a| \leq
  b$ iff $-b \leq a \leq b$ and $0 \leq b$.
\end{lemma}
\begin{proof}
  For the first statement, assume that $a \leq b$. Then $\lo a \leq
  \hi b$ whence $\lo{-b} = -\hi b \leq - \lo a = \hi{-a}$. If $-b
  \leq -a$ then $a = --a \leq --b = b$.

  For the second statement, 
  first suppose that $|a| \leq b$. We show that $- \hi b = \lo{-b}
  \leq \hi a$ and $\lo a \leq \hi b$. We distinguish three cases. If $0
  \leq \lo a$, then $a = |a| \leq b$ whence $\lo a \leq \hi b$ and
  $-\hi b \leq - \lo a \leq 0 \leq \lo a \leq \hi a$. Now suppose
  that $\lo a \leq 0 \leq \hi a$. Then $[0, \max \lbrace -\lo a, \hi
  a \rbrace ] = |a| \leq b$ so that $0 \leq \hi b$. We obtain $-\hi
  b \leq 0 \leq \hi a$ and $\lo a \leq 0 \leq \hi b$. Finally assume
  that $\hi a \leq 0$. Then $[-\hi a, -\lo a] = |a| \leq b$ whence
  $-\hi a \leq \hi b$. Then $-\hi b \leq \lo a$ and $\lo a \leq \hi
  a \leq 0 \leq - \lo a \leq \hi b$. In all three cases we have that
  $0 \leq \lo{|a|} \leq \hi b$ so that $0 \leq b$.

  Now assume that $-b \leq a$, $a \leq b$ and $0 \leq b$. We show that $|a| \leq
  b$, again by distinguishing three cases. If $0 \leq \lo a$, we
  have $|a| = a \leq b$. Similarly, if $\hi a \leq 0$, we have $-b \leq
  a$ and therefore $|a| = -a \leq --b = b$ using the first
  statement. Now assume that $\lo a \leq 0 \leq \hi a$. Then $|a| =
  [0, \max \lbrace -\lo a, \hi a \rbrace] \leq b$ if $0 \leq \hi b$
  which is precisely our assumption that $0 \leq b$.
\end{proof}

\begin{remark}
  It is in general false that $|a|\leq b$ whenever $-b \leq a \leq
  b$. To see this, let $a= [-2, 1]$ and $b = -1 = [-1, -1]$. Then
  $-b = [1, 1] \leq [-2, 1] = a$ and $a = [-2, -1] \leq [-1, -1] =
  b$. On the other hand, $|a| = [0, 2] \not\leq [-1, -1] = b$.
\end{remark}

\begin{lemma} \label{lemma:hi-lo}
Let $x, y \in \IR$. If $x \below y$ then $\lo y_n \leq \hi
x_m$ and $\lo x_m \leq \hi y_n$ for all $n, m \in \Nat$.
\end{lemma}

\begin{proof}
  Suppose that $x \below y$, let $m, n \in \Nat$ and fix a rational
  $\epsilon > 0$. Then $x_m \pm \epsilon \waybelow x_m$, hence
  there must exist $k \in \Nat$ such that $x_m \pm \epsilon
  \waybelow y_k$. By monotonicity of $y$, we may assume that $k \geq
  n, m$. But then  $\lo x_m - \epsilon = \lo{x_m - \epsilon} < \lo
  y_k \leq \hi y_k \leq \hi y_n$ as $k \geq m$. Similarly
  $\lo y_n \leq \lo y_k \leq  \hi y_k < \hi{x_m \pm \epsilon} = x_m
  + \epsilon$. The claim follows as $\epsilon$ was arbitrary.
\end{proof}

\noindent
In the following, recall that $\ell([x, y]) = y - x$ denotes
interval length, and that we identify rationals $q \in \Rat$ with
singleton intervals.

\begin{lemma} \label{lemma:fat-below}
  Let $w, x, y \in \IR$ with $w \below x$ and $w \below y$. Then $|x
  - y| \leq \ell(w_n)$ for all $n \in \Nat$.
\end{lemma}
\begin{proof}
  Let $n, k \in \Nat$. By Lemma \ref{lemma:interval-arith} it
  suffices to show that $-\ell(w_n) \leq \hi{(x -
  y)}_k = \hi x_k - \lo y_k$ and $\lo x_k - \hi y_k = \lo{(x - y)}_k \leq \ell(w_n)$.
  This is a consequence of Lemma \ref{lemma:hi-lo}, since
  \[ - \ell(w_n) = \lo w_n - \hi w_n \leq \hi x_k - \hi w_n \leq
  \hi x_k - \lo y_k \]
  and similarly,
  \[ \lo x_k - \hi y_k \leq \hi w_n - \hi y_k \leq \hi w_n - \lo w_k
  \]
  as required.
\end{proof}

\begin{lemma}\label{lemma:ext}
Let $x = (x_n)_n \in \IR$. Then $x \pm \ell(x_n) \below x$ for all
$n \in \Nat$.
\end{lemma}
\begin{proof}
Let $\alpha \in \IQ$ and $k \in \Nat$ with $\alpha \waybelow x_k \pm
\ell(x_n)$. We show that $\alpha \waybelow x_n$. By Lemma
\ref{lemma:hi-lo} we have $\lo x_k \leq \hi x_n$ as $x \below x$.
Therefore
$\lo \alpha < \lo x_k - (\hi x_n - \lo x_n) \leq \hi x_n  - \hi x_n
+ \lo x_n = \lo x_n$. Analogously one shows that $\hi x_n < \hi
\alpha$ so that $\alpha \waybelow x_n$.
\end{proof}

\begin{lemma} \label{lemma:abs-below}
  Let $x \in \IR$, $y \in \Real$ be an interval real and $\epsilon \in
  \Rat_{> 0}$.
  Then $x \pm \epsilon \below y$ whenever $|x - y| \leq
  \epsilon$.
\end{lemma}

\begin{proof}
  Let $n \in \Nat$, $a \in \IQ$ such that $ a \waybelow (x \pm
  \epsilon)_n$. We show that there must exist $k \in \Nat$ such that
  $a \waybelow y_k$. By characterisation of $\waybelow$ on $\IQ$,
  i.e. Lemma \ref{lemma:IQ-waybelow}, there exists $\delta \in
  \Rat_{>0}$ such that $a \waybelow (x \pm \epsilon)_k \pm \delta$.
  As $y$ is an interval real, there must exist $k \in \Nat$ such
  that $\ell(y_k) \leq \delta$. We may assume that $k \geq n$. As $|x -
  y| \leq \epsilon$, we have that
  \[ -\epsilon \leq \hi{(x - y)}_k = \hi x_k - \lo y_k \mbox{ and }
    \lo x_k - \hi y_k = \lo{(x-y)}_k \leq \epsilon. \]
  We obtain that
  \[ \lo a < \lo x_n - \epsilon - \delta \leq \lo x_k - \epsilon -
  \delta \leq  \hi y_k - \delta \leq \hi y_k - w(y_k) = \lo y_k \]
  and one analogously establishes that $\hi y_k < \hi a$ so that $a
  \waybelow y_k$.
\end{proof}

\noindent
We characterise Scott continuous total functions in terms of their
action on total reals.

\begin{defn}
Let $f: \Real \to \Real$ and $x \in \Real$. A \emph{modulus of
non-discontinuity} of $f$ at $x$ is
a sequence $(\delta_n,
\epsilon_n)_n$ in $\Rat_{>0} \times (\Rat_{\geq 0} \cup \lbrace
\infty \rbrace)$ such that 
\begin{enumerate}
  \item $\forall y \in \Real \sep |x-y| \leq \delta_n \to |f(x) -
  f(y)| \leq \epsilon_n$ whenever $n \in \Nat$
  \item $(\epsilon_n)_n$ is a classical null-sequence.
\end{enumerate}
We say that $f$ is \emph{not discontinuous at $x$} if there exists a
modulus of non-discontinuity of $f$ at $x$. A \emph{modulus
of intensional non-discontinuity} for $f$ is a function
$\omega: \IQ \to \Rat_{>0}$ such
that $(\ell(x_n), \omega(x_n))_n$  is a modulus of non-discontinuity
at $x$ whenever $x \in \Real$ is total. A function $f$ is
\emph{intensionally non-discontinuous} if there exists a modulus of
intensional non-discontinuity for $f$. The function $f$ is
\emph{continuous} at $x$ if, for all $\epsilon \in \Rat_{> 0}$
there exists $\delta \in \Rat_{>0}$ such that for all $y$ with $|x -
y| \leq \delta$ we have $|f(x) - f(y)| \leq \epsilon$.
\end{defn}

\begin{remark}
  It is evident that every function that is intensionally 
  non-discontinuous is non-discontinuous at every $x \in \Real$. 
  The converse is not necessarily due to the uniformity requirement. 
\end{remark}

\begin{lemma}
  If $f: \Real \to \Real$ is continuous at $x$, then $f$ is
  non-discontinuous at $x$.
\end{lemma}

\begin{proof}
  As $f$ is continuous at $x$, there exists a function $\delta:
  \Rat_{>0} \to \Rat_{\geq 0}$ such that
  $\forall n \in \Nat \sep \forall y \in \Real \sep |y - x| \leq
  \delta(2^{-n}) \to |f(x) - f(y)| \leq 2^{-n}$. Hence the sequence
  $(\delta(2^{-n}), 2^{-n})_n$ is a modulus of non-discontinuity for
  $f$ at $x$.
\end{proof}

\begin{lemma}
  If every (total) function $f: \Real \to \Real$ that is
  non-discontinuous at $0$ is in fact continuous at $0$, then
  Markov's principle holds.
\end{lemma}
\begin{proof}
  Let $P(n)$ be a decidable predicate on natural numbers and assume
  that $\neg \forall n \in \Nat \sep \neg P(n)$. We show that there
  exists $n \in \Nat$ such that $P(n)$ under the assumption that
  every function that is non-discontinuous at $0$ is in fact
  continuous at $0$.
  Consider, for $n \in \Nat$, the function $f_n: \IR \to \IR_\bot$
  defined by
  \[ f_n (x) = \begin{cases}
    \bot &  \forall n \leq k \sep \neg P(k)  \\
    x \cdot \min \lbrace k \leq n \mid P(n) = 1 \rbrace & \exists k
    \leq n \sep P(k)
  \end{cases}\]
  and let $f = \sqsup_n f_n$.
  Clearly $(f_n)_n$ is monotone so that $\sqsup_n f_n$ exists, and
  $\sqsup_n f_n$ defines a total function with $f(0) = 0$. We claim
  that the sequence $(\epsilon_n, \delta_n)$ defined by $\delta_n =
  \frac 1{n^2}$ and
  \[ \epsilon_n = \begin{cases}
    \infty & \forall k \leq n \sep \neg P(k) \\
    \frac 1n & \exists k \leq n \sep P(k)
  \end{cases}\]
  is a modulus of non-discontinuity for $f$ at $0$. To see this, let
  $n \in \Nat$. If $\forall k \leq n \sep \neg P(n)$, then
  $\epsilon_n = \infty$ and there is nothing to show. So assume that
  $\exists k \leq n \sep P(n)$ and $|y - 0| \leq \delta_n = \frac
  1{n^2}$. Then $|f(y) - f(0)| = |f(y)| = |y| \cdot \min \lbrace k \leq n
  \mid P(n) \rbrace \leq |y|\cdot n \leq \delta_n \cdot n = \frac
  1{n^2} \cdot n = \frac 1n = \epsilon_n$.

  By assumption, $f$ is also continuous at $0$, hence for $\epsilon
  = 1$ there exists $\delta$ such that for all $y$ with $|y| = |y - 0|
  \leq \delta$ we have $|f(y)| = |f(y) - f(0)| \leq 1$. As $\delta \in
  \Rat_{>0}$ there exists $n$ such that $n \geq \frac{1}\delta$. We
  claim that there exists $k \leq n$ such that $P(k)$. As this
  property is decidable, we may assume that 
  $\neg P(k)$ for all $k \leq n$, and establish a contradiction to
  prove the claim. By assumption, we have $\neg\forall k \neg P(k)$
  so that it suffices to show that $\forall k \sep \neg P(k)$ to
  establish the desired contradiction. So let
  $k \in \Nat$, we need to show $\neg P(k)$ to prove the claim.
  If $k \leq n$ then $\neg P(k)$ is given. Now let $k > n$. As
  $P$ is decidable, we may assume $P(k)$ and establish a
  contradiction to show that $\neg P(k)$. So assume that $P(k)$.
  Then $f(y) = \min \lbrace k' \leq k \mid P(k') \rbrace \cdot x$.

  We now use that for $y = \delta$, we have that $|fy| = \min
  \lbrace k' \leq k \mid P(k')  \rbrace \cdot \delta \leq 1$ so
  that $\min \lbrace k' \leq k \mid P(k')  \rbrace \leq \frac
  1\delta \leq k$ so that there exists $k' \leq k$ with $P(k)$,
  contradiction.
\end{proof}

\noindent
The following lemma shows that equality is stable under adding
classical null-sequences. Again, this is a consequence of the
classical formulation of equality on reals.

\begin{lemma} \label{equality}
Let $x = (x_n) \in \IR$ and let $(q_n)$ be a classical null
sequence. If $y = (x_n \pm q_n)$, then $x = y$.
\end{lemma}
\begin{proof}
We establish that $(x_n \pm q_n)_n \below (x_n)_n$ and $(x_n)_n
\below (x_n \pm q_n)_n$. The first relation is immediate using the
Definition \ref{defn:cont-compl}, for if $\alpha \waybelow x_n \pm
q_n$, we have $\alpha \waybelow x_n \pm q_n \below x_n$ whence
$\alpha \waybelow x_n \pm q_n$ by Lemma \ref{lemma:way-below-below}.
For the second (converse) relation, assume that $\alpha \waybelow
x_n$. Using Lemm \ref{lemma:IQ-waybelow}, there exits $\epsilon \in
\Rat$ with $\epsilon > 0$ and $\alpha \waybelow x_n \pm \epsilon$.
As $(q_n)_n$ is a classical null sequence, there must exist $m \in
\Nat$ such that $q_m < \epsilon$. For $k = \max \lbrace n, m
\rbrace$ we therefore obtain that
$\alpha \waybelow x_n \pm \epsilon \below x_k \pm q_k$ as required.
\end{proof}

\begin{lemma}
Suppose that $f: \IR \to \IR$ is Scott continuous and total. Then $f
\restriction \Real$ is intensionally non-discontinuous. 
\end{lemma}

\begin{proof}
Fix an antitone function $b: \Rat_{>0} \to \Nat$ such that
$2^{-b(q)} \leq q$ and let $s(\alpha) = b(\ell(\alpha))$ for $\alpha
\in \IQ$ with $\ell(\alpha) > 0$. Then $s$ is antitone with respect
to interval length, i.e. $0 < \ell(\alpha) \leq \ell(\beta)$ implies
$s(\alpha) \geq s(\beta)$ for all $\alpha, \beta \in \IQ$. 

By the approximation lemma (Lemma \ref{lemma:sup-simple}) we have
that $f = \sqsup_n f_n$ is a supremum of step functions.
Define $\omega: \IQ \to \Rat_{\geq 0} \cup \lbrace \infty \rbrace$ by
$\omega(\alpha) = \ell(f_{s(\alpha)}(\alpha \pm \ell(\alpha)))$. We show that
$\omega$ is a modulus of intensional non-discontinuity for $f$. 

Let $x \in \Real$ be given. We first show that $\omega(x_n)_n$ is a
classical null-sequence.  Let $y_n = x_n \pm \ell(x_n)$. Then $x =
y$ by Lemma \ref{equality}  and $y = \sqsup_n y_n$ by Lemma
\ref{lemma:own-sup}. Then $f(x) = f(\sqsup_m y_m) = \sqsup_m f(y_m) =
\sqsup_m \sqsup_n f_n(y_m) = \sqsup_n f_n(y_n)$ by Corollary
\ref{cor:diagonal-sup}, and applying Lemma \ref{lemma:own-sup} this
gives that $f(x) = (f_n(y_n))_n$ as a sequence. 
Now let $\epsilon > 0$. As $f(x)$ is total
and $f(x) = (f_n(y_n))_n$,  there must exist $k \in \Nat$ such that
$\ell(f_k(y_k)) \leq \epsilon$. By definition of $s$, there moreover
must exist $i \in \Nat$ such that $s(x_i) \geq k$. As $s$ is
antitone with respect to interval length and $(y_n)_n$ is
increasing, we may assume that $i \geq k$. Then $\omega(x_i) =
\ell(f_{s(x_i)} (y_i)) \leq \ell(f_k (y_i)) \leq \ell(f_k(y_k))
\leq \epsilon$ as required.
 
 We now show that $(\ell(x_n), \omega(x_n))_n$ is a modulus of
 non-discontinuity of $f$ at $x$. So let $n \in \Nat$ and fix $y \in
 \Real$ with $|x - y| \leq \ell(x_n)$. Then
 $x_n \pm \ell(x_n) \below x \pm \ell(x_n) \below y$ by Lemma
 \ref{lemma:abs-below}, and therefore $f(x_n \pm \ell(x_n)) \below
 f(y)$. As also $x_n \pm \ell(x_n) \below x \pm \ell(x_n) \below x$
 we have, again by monotonicity of $f$, that $f(x_n \pm \ell(x_n))
 \below f(y)$. As $f_{s(x_n)} \below f$, this gives
 $f_{s(x_n)} (x_n \pm \ell(x_n)) \below f(x)$ and $f_{s(x_n)} (x \pm
 \ell(x_n)) \below f(y)$. By Lemma \ref{lemma:fat-below} we may
 conclude that $|f(x) - f(y)| \leq \ell(f_{s(x)}(x \pm \ell(x_n)) =
 \omega(x_n)$ which finishes the proof.
\end{proof}

\begin{thm}
  Suppose that $f: \Real \to \Real$ is intensionally non-discontinuous.
  Then there exists a Scott continuous function $g: \IR \to \IR\lift$
  such that $g \restriction \Real = f$. 
\end{thm}

\begin{proof}
Let $\omega$ be a modulus of intensional non-discontinuity for $f$ and
suppose that $\IQ = \lbrace \alpha_i \mid i \in \Nat \rbrace$ is an
enumeration of $\IQ$. For $\alpha \in \IQ$ let $m(\alpha) = \frac 12
(\lo \alpha + \hi \alpha)$ denote the midpoint of $\alpha$. Put 
\[ g_n = \sqsup_{i \leq n} \alpha_i \step f(m(\alpha_i))_n \pm
\omega(\alpha_i) \]
and let $g = \sqsup_n g_n$. 

We first show that $g_n$ is well-defined for all $n \in \Nat$, i.e.
satisfies the consistency requirement for step functions. To see
this, let $I \subseteq \Nat$ be finite and suppose that $\lbrace
\alpha_i \mid i \in I \rbrace$ is consistent. We show that
$f(m(\alpha_i))_k \pm \omega(\alpha_i)$ are consistent for arbitrary
$k \in \Nat$. The claim follows for $k = \max I$.  By consistency of
$\lbrace \alpha_i \mid i \in I \rbrace$ we obtain $q \in \Rat$ such
that $\alpha_i \below q$ for all $i \in I$, e.g. $q = \lbrace \max
\lo \alpha_i \mid i \in I \rbrace$. We now obtain that
$\alpha_i = m(\alpha_i) \pm \frac 12 \ell(\alpha_i) \below q$ and
$\alpha_i \below m(\alpha_i)$, hence by Lemma \ref{lemma:fat-below}
we obtain $|q - m(\alpha_i)| \leq \frac 12 \ell(\alpha_i)$. As
$\omega$ is a modulus of intensional non-discontinuity for $f$, this
gives $|f(q)  - f(m(\alpha_i)| \leq \omega(\alpha_i)$. By Lemma \ref{lemma:abs-below} this gives $f(m(\alpha_i))
\pm w(\alpha_i) \below f(q)$. As $f(m(\alpha_i))_k \pm
\omega(\alpha_i) \pm 2^{-n} \waybelow f(m(\alpha_i))_k \pm
\omega(\alpha_i)$ and $f(m(\alpha_i)) \pm \omega(\alpha_i) \below
f(q)$, there must exist $r(i)$ such that $f(m(\alpha_i))_k \pm
\omega(\alpha_i) \pm 2^{-n} \waybelow f(q)_{r(i)}$. For $r = \max
\lbrace r(i) \mid i \in I \rbrace$ we therefore have
$f(m(\alpha_i)) \pm w(\alpha_i) \pm 2^{-n} \waybelow f(q)_r$, that
is, the set $\lbrace f(m(\alpha_i))_k \pm \omega(\alpha_i) \pm
2^{-n} \mid i \in I \rbrace$ is consistent for all $n \in \Nat$. As
consistency on $\IQ$ is continuous by Lemma \ref{lemma:IQ-cons-cont}
this shows that $\lbrace f(m(\alpha_i))_k \pm \omega(\alpha_i) \mid
i \in I \rbrace$ is consistent.

We now demonstrate that $g$ preserves total reals, that is, $g(x)$
is total whenever $x$ is.
Let $x \in \Real$ be total and $y_n = x_n \pm 2^{-n}$. Then, for
every $n \in \Nat$ there must exist
\begin{itemize}[label=$\triangleright$]
  \item $i \in \Nat$ such that $\omega(y_i) \leq \frac{\epsilon}4$
  \item $j \in \Nat$ such that $y_i \in \lbrace \alpha_i \mid i \leq
  j \rbrace$
  \item $k \in \Nat$ such that $\ell(f(m(y_i))_k) \leq
  \frac{\epsilon}2$
\end{itemize}
so that for $l = \max \lbrace i, j, k \rbrace$ we have:
\begin{align*}
  g(x)
  & \above g_l(x_i) 
  && \mbox{(as $g_l = \below \sqsup_n g_n = g$ and $x_i \below
  \sqsup_n x_n = x$)} \\
  & \above (y_i \step f(m(y_i))_l \pm \omega(y_i)) \, (x_i) 
  && \mbox{(as $y_i \in \lbrace \alpha_i \mid i \leq j \rbrace
  \subseteq \lbrace \alpha_i \mid i \leq k \rbrace$)} \\
  & = f(m(y_i))_l \pm \omega(y_i) 
  && \mbox{(definition of step functions)} \\
  & \above f(m(y_i))_k \pm \omega(y_i)
  && \mbox{(monotonicity of $f(m(y_i))$ and $k \leq l$)}
\end{align*}
In particular, this gives that
\[ \ell(g(x)) \leq \ell(f(m(y_i))_l \pm \omega(y_i)) \leq
\ell(f(m(y_i))_l) + 2 \omega(y_i) \leq \frac\epsilon 2 + 2 \frac
\epsilon 4 = \epsilon \]
as required.

We finally demonstrate that  $g(x) \below f(x)$ for $x \in \Real$. As $g(x) =
\sqsup_n g_n(x_n)$, it suffices to show that $g_n(x_n) \below f(x)$.
So let $n \in \Nat$. Then
$g_n(x) = \sqsup \lbrace f(m(\alpha_i))_n \pm \omega(\alpha_i) \mid 0
\leq i \leq n, \alpha_i \waybelow x_n \rbrace$ so that the claim
follows once we show that $f(m(\alpha)_n \pm \omega(\alpha) \below
f(x)$ whenever $\alpha \in \IQ$ with $\alpha \waybelow x_n$. So assume
that $\alpha \waybelow x_n$ so that in particular $\alpha \below
x$. As also $\alpha \below m(\alpha)$ we have that $|x - \alpha|
\leq \ell(\alpha)$ by Lemma \ref{lemma:fat-below}. As $\omega$ is a
modulus of intensional non-discontinuity, this gives $|f(x) -
f(m(\alpha))| \leq \omega(\alpha)$ and in turn, using Lemma
\ref{lemma:abs-below} that $f(m(\alpha)) \pm \omega(\alpha) \below
f(x)$. Using monotonicity of $f(m(\alpha))$ we obtain $f(m(\alpha))_n
\pm \omega(\alpha) \below f(x)$ as required.

We now finish the proof by showing that $f(x) = g(x)$ whenever $x
\in \Real$ is total. But this follows from $g(x)$ being maximal by
Lemma \ref{lemma:total-max} and the fact that $g(x) \below f(x)$.
\end{proof}

\section{Conclusion and Discussion}

The main guiding principle of our development here was
``constructive existence with classical correctness''. The main goal
was to constructively rationalise standard practice in constructive
analysis: constructions are carried out in the universe of classical
mathematics, and then a secondary argument is used to show that they
are in fact effective. This is reflected in our approach that
emphasises constructive existence, but contends itself with
classical correctness arguments. One consequence of this is that
correctness assertions have no computational content under a
realisability interpretation. While this can also be achieved by
achieved using different methods (e.g. non-computational quantifiers
\cite{Berger:1993:PEN} or $\mathrm{Prop}$-valued assertions in the
calculus of constructions \cite{Coquand:1988:CC}), we consciously
took a pragmatic approach that aligns with computable analysis. As
next step, our approach should be benchmarked both mathematically
(e.g. by establishing standard results of computable analysis as
carried out e.g. in \cite{Schwichtenberg:2008:RIP}) and
experimentally, by implementing our theory in a theorem prover such
as Coq \cite{Bertot:2004:ITP} or Minlog \cite{Schwichtenberg:2012:PC}.

\section*{Acknowledgments}
\noindent The authors wish to express their gratitude to Helmut
Schwichtenberg, Dieter Spreen and Andrej Bauer for continued
discussions about the topic, and to the anonymous reviwers for their
constructive suggestions to improve the paper. We also acknowledge the support of a
visiting scholarship by the Humboldt Foundation.

\bibliographystyle{alpha}
\bibliography{refs}
\end{document}